\documentclass[twocolumn]{autart}    % Enable this line and disable the 
\let\theoremstyle\relax

\usepackage{graphicx}     

\usepackage{amsmath, amssymb,bbm,amsthm,bbold}
\usepackage{enumitem}
\usepackage{xcolor}
\usepackage{comment}
\usepackage{natbib}
\usepackage{algorithm}
\usepackage{wrapfig}

\usepackage{algpseudocode}
\usepackage{hyperref}
\hypersetup{
    colorlinks=true,
    linkcolor=blue,
    urlcolor=blue,
    citecolor=blue
}
\usepackage[english]{babel}
\usepackage{epstopdf}
\usepackage{tikz}
\usepackage{chngcntr}
\usepackage{apptools}
\usepackage{calrsfs}  % molto riccioluto ma BELLO

\newtheorem{theorem}{Theorem}
\newtheorem{remark}{Remark}

\newtheorem{lemma}[theorem]{Lemma}

\newtheorem{proposition}[theorem]{Proposition}
\newtheorem{corollary}[theorem]{Corollary}

\newtheorem{assumption}{Assumption}

\allowdisplaybreaks %per dividere in maniera automatica le equazioni

\begin{document}
%%%%%%%%%%%%%%%%%%%%%%%%%%%%%%%%%%%%%%%%%%%%%%%%%%%%%%%%%%%%%%%%%%%%%%%%%%%%%%%%%%
\begin{frontmatter}
\title{
Distributed Observer Design for Discrete-Time LTI Systems via Jordan Canonical Form}

\author[unipd]{Giulio Fattore}\ead{fattoregiu@dei.unipd.it},    % Add the 
\author[unipd]{Maria Elena Valcher}\ead{meme@dei.unipd.it},  
\author[SKLSAPI]{Rui Gao}\ead{gaorui@ise.neu.edu.cn},  
\author[SKLSAPI]{Guang-Hong Yang}\ead{yangguanghong@ise.neu.edu.cn.},

   \address[unipd]{Department of Information Engineering, University of Padova, Via Gradenigo 6/B, 35131 Padova, Italy} 
 \address[SKLSAPI]{State Key Laboratory of Synthetical Automation for Process Industries, Northeastern University, Shenyang 110819, China}


\begin{abstract}
This paper addresses the problem of distributed state estimation for discrete-time linear time-invariant systems. Building on the framework proposed in \cite{GaoYang}, we exploit the Jordan canonical form of the system matrix to develop two distributed estimation schemes that ensure asymptotic convergence of local estimates to the true system state. In both approaches, each node reconstructs the components of the state that are locally detectable for it via a Luenberger observer, while employing a consensus-based mechanism to estimate the components that are not directly detectable.

The first scheme relies on local observers whose dimension matches that of the original state vector; however, its applicability requires the satisfaction of a large set of inequalities. The second scheme, in contrast, can be implemented under less restrictive conditions, but results in observers of increased (augmented) order.
For both methods, we derive necessary and sufficient conditions - expressed in terms of the eigenvalues of the system matrix and certain submatrices of the communication network Laplacian - that guarantee the existence of a distributed observer achieving asymptotically accurate estimation. Compared to  \cite{GaoYang}, the proposed approaches offer greater flexibility in the selection of coupling gains and impose less stringent solvability conditions.
\end{abstract}
\end{frontmatter}

The last decades have witnessed a significant  interest in distributed state estimation, largely driven by the limitations of centralized architectures in modern large-scale and networked systems. In many applications, such as sensor networks, cyber-physical systems, and power grids, collecting all measurements at a single location is often impractical due to communication constraints, energy limitations, and scalability issues. Distributed approaches overcome these challenges by allowing agents/nodes/sensors to process data locally and cooperate through limited information exchange. 
In the typical distributed estimation set-up, a dynamical system is observed by a network of sensors, each having access only to partial measurements of the system state and, in some cases, limited knowledge of the control inputs. As individual sensors are typically incapable of reconstructing the full state independently, they must cooperate by exchanging information with neighboring nodes. This interaction is commonly implemented through consensus-based strategies, enabling the network to progressively refine local estimates and achieve an accurate reconstruction of the global system state.

A wide range of distributed state estimation algorithms have been developed over the years. In the stochastic setting, several approaches are based on variants of the Kalman filter \cite{Tomlin,OS_distrKalman2,OS_distrKalman}, while in the deterministic setting, Luenberger-type observers are commonly employed. The continuous time case has been extensively studied under various assumptions on network connectivity and joint detectability \cite{Han2019,yang2022state}. However, the problem in discrete-time is significantly more challenging, {\color{black} as 
 high-gain consensus mechanisms adopted for   continuous-time systems do not directly extend to discrete-time ones} \cite{Wang24}.

Distributed state estimation for discrete-time autonomous systems has been investigated by Park and Martins in a series of papers \cite{Nuno1, Nuno2, Nuno3}   using networks of Luenberger-like observers. In their framework, each observer has an augmented  state, processes a subset of the system outputs, and updates its estimate using both local measurements and information received from neighboring nodes. The use of an augmented state allows to recast  the observer design problem as a decentralized dynamic output feedback control problem, allowing the use of classical control-theoretic tools. This approach has been further extended in \cite{Disaro_TAC2025} to systems affected by unknown disturbances and partial input measurements.
More recently, \cite{Wang24} proposed a unified framework for distributed state estimation applicable to both continuous- and discrete-time linear systems, including cases with time-varying communication graphs. The approach relies on partitioning the system spectrum into subsets associated with observable and unobservable subspaces. The estimation scheme then employs high-gain techniques in continuous time, while in the discrete-time case it resorts to two appropriately chosen sampling rates.

 The paper \cite{MitraSundaram} explores two distinct strategies for distributed observer design.  The first one assumes detectability within each source component of the communication graph and introduces the concept of {\em multisensor observable canonical decomposition}. This decomposition transforms the system and (overall) output matrices into a lower block-triangular form, where each diagonal block pair corresponds to an observable subsystem, thus enabling each agent to autonomously estimate a specific portion of the state.
The second strategy in \cite{MitraSundaram} assumes that, for every unstable eigenvalue of the system matrix, there exists at least one root node in each source component for which that eigenvalue is observable. Based on this assumption, a coordinated Jordan canonical transformation - arguably the first use of such a decomposition in distributed observer design - is applied, followed by a local transformation into Kalman detectability form. The resulting estimation scheme relies on observers of augmented order. In both approaches, the edge weights are treated as design parameters rather than network constraints. This framework was later extended in \cite{MitraSundaram2019} to account for adversarial agents, enhancing robustness against compromised nodes.
In \cite{FioravantiTAC2024}, a distributed state estimation scheme based on a Luenberger-type observer is proposed, which guarantees exact state reconstruction in a finite number of steps. A key advantage of this approach is that it enables agents to design the observer gains in a fully distributed manner, {\color{black} but this} comes at the cost of stronger assumptions, i.e., the joint observability of the sensor network and the requirement that, at each time step, the sensor nodes have access to the average of the state estimates across the entire network.

Building on the use of Jordan canonical forms, %Gao and Yang 
\cite{GaoYang} proposed a distributed observer design in which each node applies a local (and hence fully distributed) permutation that brings the system into Kalman detectability form while preserving the Jordan structure. Each agent employs a Luenberger-type observer to estimate locally detectable state components and uses a consensus-based mechanism, with a common coupling gain, to reconstruct the remaining components. This formulation leads to necessary and sufficient conditions - expressed in terms of the eigenvalues of the system matrix and the communication graph - for the solvability of the distributed estimation problem, without modifying the network topology or edge weights.
For a recent survey  of distributed observer design methods for linear time-invariant systems, either of deterministic or of  stochastic nature, the reader is referred to \cite{ReviewDIstrObs,Rego19}.
%\smallskip

Building upon the results in \cite{GaoYang}, this paper advances the state of the art in distributed state estimation by proposing two design methodologies based on the Jordan canonical form of the system matrix. The first approach employs local observers whose dimension matches that of the original state vector; however, its applicability requires relatively stringent conditions relating the eigenvalues of the system matrix to those of the communication network Laplacian. The second approach, by contrast, can be implemented under less restrictive conditions on the same spectral quantities, at the expense of increasing the observer order through state augmentation. For both methods, we establish necessary and sufficient conditions that guarantee the existence of a distributed observer achieving asymptotically accurate state estimation.

Compared with the design strategy in \cite{GaoYang}, the proposed methods do not impose constraints on each individual component of the estimation error. Instead, they preserve the structural partitioning induced by the Jordan decomposition, thereby significantly reducing the number of conditions that must be verified. Furthermore, rather than relying on a single coupling gain for all state components, the proposed framework introduces distinct gains for each Jordan miniblock, leading to weaker feasibility requirements and, consequently, less restrictive solvability conditions. A detailed comparison with \cite{GaoYang} is provided in Section \ref{sec.comparison}.
%\smallskip

The paper is organized as follows. Section \ref{sec.pb.setup} formalizes the problem setup.
In section \ref{Sec.Reordering}  a preliminary reordering of the Jordan miniblocks corresponding to unstable eigenvalues (and of the corresponding state entries) is performed for each single agent, with the goal of
identifying  the blocks of the state vector that are completely observable, partially observable, and unobservable for that agent. Sections \ref{sec.strategy1} and \ref{sec.strategy2} present the first and the second approach, respectively, including necessary and sufficient conditions for problem solvability. 
Section \ref{sec.comparison} compares the two strategy to each other and to the original solution proposed in \cite{GaoYang}.
 Section \ref{sec.example} provides an illustrative example that demonstrates the effectiveness of the proposed approach. Finally, Section \ref{sec.conclusion} draws the main conclusions.
This paper builds upon the conference paper \cite{Fattore_ECC2026}, with
which it shares the content of Sections \ref{sec.pb.setup}, \ref{Sec.Reordering}, \ref{sec.strategy2} and part of Section \ref{sec.comparison}. Sections \ref{sec.strategy1}, \ref{sec.und}, the illustrative example in Section \ref{sec.example} and the Appendix are original.
%\smallskip

{\bf Notation.} 
The sets of real numbers, nonnegative integers, and complex numbers are denoted by $\mathbb{R}$, $\mathbb{Z}_+$, and $\mathbb{C}$, respectively.
Given two integers $h$ and $k$, with $h \le k$, we let $[h,k]$ denote the set $\{h,h+1,\dots,k\}$.
The symbols $I$ and $\mathbb{0}$ denote the identity matrix and   the zero matrix/vector of suitable dimensions, respectively.  The $i$th {\em canonical vector} whose entries are all zeros except for the $i$th, which is equal to $1$, is denoted by $\mathbb{e}_{i}$ and its dimension can be deduced from the context. 
A square matrix $P\in {\mathbb R}^{n \times n}$ obtained from the identity $I_n$ via row-column permutations 
is called {\em permutation matrix}.
Given any matrix $Q$, its $(i,j)$th entry is denoted by $[Q]_{i,j}$.  
The {\em Kronecker product} is denoted by $\otimes$.
Given matrices $M_i, i\in [1,p]$, the  column stacking and the row juxtaposition of these matrices are denoted by ${\rm col}\{M_i\}_{i\in[1,p]}$,  and ${\rm row}\{M_i\}_{i\in[1,p]}$, respectively, while 
the block diagonal matrix whose $i$th diagonal block is the matrix $M_i$ is denoted by ${\rm diag} \{M_i\}_{i\in[1,p]}$. Given two sets ${\mathcal S}_1$ and ${\mathcal S}_2$, the {\em difference set} ${\mathcal S}_1 \setminus  {\mathcal S}_2$ is the set of elements of ${\mathcal S}_1$ that are not in ${\mathcal S}_2$. Given a set ${\mathcal S}$, the symbol ${\rm card}({\mathcal S})$ denotes its {\em cardinality}.

A {\em weighted, directed graph} (digraph) is   a triple ${\mathcal G} = ({\mathcal V}, {\mathcal E}, {\mathcal A})$, where ${\mathcal V} = \{1, \dots, N\}=[1,N]$ is the set of nodes, ${\mathcal E}\subset {\mathcal V}\times {\mathcal V}$  is the set of edges, and ${\mathcal A}\in \mathbb{R}^{N\times N}$ is the nonnegative, weighted {\em adjacency matrix} which satisfies $[{\mathcal A}]_{i,j}>  0$ if and only if $(j,i)\in {\mathcal E}$. We assume that $[{{\mathcal A}}]_{i,i} =0$ for every $i\in {\mathcal V}$. For every $i\in [1,N],$  the set of {\em neighbors of node $i$} is  ${\mathcal N}_i \doteq \{j\in [1,N]: (j,i) \in {\mathcal E}\}$. Note that we assume that $i\not\in {\mathcal N}_i$.
The  {\em in-degree} of  node $i$ is $d_i \doteq \sum_{j=1}^N[{{\mathcal A}}]_{i,j} = \sum_{j\in {\mathcal N}_i}[{{\mathcal A}}]_{i,j}$, while the {\em in-degree matrix} $\mathcal{D}$ is the diagonal matrix defined as $\mathcal{D}\doteq {\rm diag}\{d_i\}_{i\in [1,N]}$. The {\em Laplacian}  associated with ${\mathcal G}$
is defined as ${\mathcal{L}}\doteq \mathcal{D}-\mathcal{A}$. 
A {\em weighted undirected graph} is a weighted graph ${\mathcal G} = ({\mathcal V}, {\mathcal E}, {\mathcal A})$ whose adjacency matrix ${\mathcal A}$ is symmetric. When so, also the Laplacian is symmetric and is a positive semi-definite matrix.
The {\em spectrum} of a square matrix $A$, denoted by $\sigma(A)$, is the set of all its eigenvalues.
We let $r$ be the number of distinct eigenvalues of $A$, so that $\sigma(A) = \{\lambda_1, \dots, \lambda_r\}$.  The {\em algebraic multiplicity} and the {\em geometric multiplicity} of the eigenvalue $\lambda_\ell\in {\mathbb C}$  are denoted by $a_\ell$ and $g_\ell$, respectively. 
A matrix $A\in  {\mathbb C}^{n \times n}$ is in {\em Jordan form} \cite{HornJohnson} if it takes the form
\begin{equation}\label{Jform}
    A = {\rm diag} \{A^\ell\}_{\ell\in [1,r]},\end{equation}
where $A^\ell$ is the {\em Jordan block} of size $a_\ell \times a_\ell$ associated  with the eigenvalue $\lambda_\ell$, by this meaning that
\begin{equation}\label{Jblock}
A^\ell = {\rm diag} \{A^{\ell,h_\ell}\}_{h_\ell\in [1,g_\ell]}\in {\mathbb C}^{a_\ell \times a_\ell},\end{equation}
and $A^{\ell, h_\ell}$ is the $h_\ell$th {\em Jordan miniblock} associated with $\lambda_\ell$ and hence taking the form
\begin{equation}\label{Jminiblock}
A^{\ell, h_\ell} = \begin{bmatrix} \lambda_\ell & 1 & 0 &\dots & 0\cr   
0 & \lambda_\ell & 1 & \dots &0\cr
\vdots &\vdots &\ddots &\ddots& \vdots \cr
0 & 0 &0 & \ddots & 1\cr
0 & 0 &0 & \dots & \lambda_\ell\end{bmatrix}\in{\mathbb C
}^{d^{\ell, h_\ell} \times d^{\ell, h_\ell}}.\end{equation}
 %\smallskip
 
\section{Problem set-up}\label{sec.pb.setup}
Consider a discrete-time linear time-invariant state-space model
\begin{align}\label{eq.sys}
    x(t+1)=Ax(t)+Bu(t),
    \qquad t\in \mathbb{Z_+},
\end{align}
where $x(t)\in\mathbb{R}^n$ is the state and $u(t)\in\mathbb{R}^m$ the control input. For simplicity, we assume that   $A$ is an $n\times n$ real matrix   {\em  in Jordan form} \eqref{Jform}-\eqref{Jblock}-\eqref{Jminiblock}, which implies, in particular, that all its $r$ distinct eigenvalues, $\lambda_1, \lambda_2, \dots, \lambda_r$,  are real.
Without loss of generality  we assume that the eigenvalues have been sorted in such a way that $|\lambda_\ell|\ge1$  for $\ell\in [1,r_u]$, while $|\lambda_\ell|<1$  for $\ell\in [r_u+1,r]$. 
\smallskip

\begin{remark}The assumption that 
$A$ is in Jordan form is not restrictive, since it can always be achieved through a change of basis. The assumption that all eigenvalues of 
$A$ are real is restrictive, but the extension to complex eigenvalues follows easily using the {\em real Jordan form}, as in \cite{MitraSundaram2019} and \cite{GaoYang}. Both assumptions are made only to simplify the notation.
\end{remark}

We assume that there is a network of $N$ sensors, connected through a directed graph $\mathcal{G}=({\mathcal V}, \mathcal{E},\mathcal{A})$, each of them taking an indirect measure $y_i(t)\in\mathbb{R}^{p_i}$ of the state:
\begin{align}\label{eq.out}
    y_i(t)=C_ix(t),\quad i\in[1,N].
\end{align}
The objective is to design, for every node $i \in [1,N]$, a local observer that produces an estimate $\hat{x}_i(t)$ of the true state $x(t)$, such that the $i$th {\em estimation error} 
    $e_i(t)\doteq x(t)- \hat{x}_i(t)$
 asymptotically goes to zero, that is,
\begin{align*}
    \lim_{t \to +\infty} e_i(t) = \mathbb{0}, \qquad \forall i \in [1,N].
\end{align*}
Of course, a necessary condition for the previous problem to be solvable is that by using (the input $u$ and) the output measurements of all the agents, it is possible to estimate the state of the system. Therefore, we introduce the following assumption.
\smallskip

\begin{assumption}\label{ass.jointdet}
\cite{GaoYang}
The $N$ agents are {\em jointly detectable}, i.e.,
$$\left(A, \begin{bmatrix} C_1 \cr \vdots \cr C_N \end{bmatrix}\right)
\ \mbox{is a detectable pair.}$$
\end{assumption}
%\smallskip
%%%%%%%%%%%%%%%%%\noindent 
According to the  block-partition of $A$,   the state splits as 
\begin{align*}
    x(t)={\rm col}\{x^\ell(t)\}_{\ell\in[1,r]},\quad x^\ell(t)={\rm col}\{x^{\ell,h_\ell}(t)\}_{h_\ell\in[1,g_\ell]},
\end{align*}
and the output matrices $C_i$, $i\in[1,N],$ as
$$C_i = \begin{bmatrix} C_i^1 & \dots &C_i^r\end{bmatrix}\in {\mathbb R}^{p_i \times n},$$
where $C_i^\ell$, $\ell\in[1,r],$ are in turn partitioned as
$$C_i^\ell = \begin{bmatrix} C_i^{\ell,1} & \dots& C_i^{\ell,g_\ell}\end{bmatrix}\in {\mathbb R}^{p_i \times a_\ell},$$
and each   block $C_i^{\ell, h_\ell}$ has size $p_i\times d^{\ell, h_\ell}$.

We now want to propose two strategies to achieve distributed state estimation, both relying on the Jordan structure of the system matrix $A$. 
In both cases, we will resort to a suitable  permutation of the matrix $A$, and hence of the corresponding entries of the state vector, that will lead to identify, for each agent $i$, a detectable and an undetectable part of the state vector. 
In the first strategy, the sum of the sizes of the detectable and the undetectable parts will coincide with the dimension $n$ of the state vector, while the second strategy will resort to an augmented state observer since the detectable and the undetectable parts will share some common entries. As we will see, the first solution is more efficient, but it works under more restrictive conditions than the other.
\\
To help the reader in keeping track of the symbols used in the paper, 
we have introduced Table 1, below.

\begin{table}[htbp]
\centering
\begin{tabular}{|c|p{6.6cm}|}
\hline
\textbf{Symbol} & \textbf{Meaning} \\
\hline
$\lambda_\ell$ & Generic eigenvalue of $A$ \\
$\alpha_\ell$ & Algebraic multiplicity of $\lambda_\ell$ \\
$g_\ell$ & Geometric multiplicity of $\lambda_\ell$ \\
$A^\ell$ & Jordan block associated with $\lambda_\ell$ \\
$A^{\ell,h_\ell}$ & $h_\ell$th Jordan miniblock associated with $\lambda_\ell$ \\
$h_\ell$ & Index of a generic Jordan miniblock associated with $\lambda_\ell$ \\
$d^{\ell,h_\ell}$ & Dimension of the $h_\ell$th Jordan miniblock associated with $\lambda_\ell$ \\
$t_i^{\ell,h_\ell}$ & Index of the first nonzero column of the block $C_i^{\ell,h_\ell}$,  when $C_i^{\ell,h_\ell}\ne {\mathbb 0}$, $i\in [1, N]$ and $
\ell\in [1,r_u]$
\\
\hline
\end{tabular}
 \vspace{0.5em}
 
    {\small Table 1. Notation and meaning of the main symbols}
\end{table}

\section{Preliminary reordering of the Jordan mini\-blocks} \label{Sec.Reordering}

As a first step, 
we act on the matrix pairs $(A,C_i), i\in [1,N],$ in order to sort, for every $i$, the blocks of $(A, C_i)$ and the entries of the vector $x(t)$ according to the detectability properties of the pair $(A,C_i)$. The procedure is close to the one proposed in \cite{GaoYang}, but introduces some  changes that lead to a simpler notation and pave the way to simpler solvability conditions. 
\medskip

{\bf Step 1: Reordering within $A^\ell$ of the Jordan miniblocks $A^{\ell, h_\ell}$ for agent $i$ based on their observability properties.}\\ For every $i\in [1,N]$ and $\ell\in [1,r_u]$,  we partition the set of indices $h_\ell \in [1, g_\ell]$ into three disjoint sets, based on the properties of the block $C_i^{\ell,h_\ell}$ in $C_i$ (corresponding to the Jordan miniblock $A^{\ell,h_\ell}$):
\begin{subequations}
\label{eq:defGil}
    \begin{align}
    \mathbb{G}_{i,1}^{\ell}&\doteq\{h_\ell\in[1,g_\ell]:C_i^{\ell,h_\ell}=\mathbb{0}\},\\
    \mathbb{G}_{i,2}^{\ell}&\doteq\{h_\ell\in[1,g_\ell]:
      (C_i^{\ell,h_\ell}\ne\mathbb{0}) \wedge\ (t_i^{\ell, h_\ell} >1) \},\\
    \mathbb{G}_{i,3}^{\ell}&\doteq\{h_\ell\in[1,g_\ell]:
    (C_i^{\ell,h_\ell}\ne\mathbb{0}) \wedge\ (t_i^{\ell, h_\ell} = 1) \}.
\end{align}
\end{subequations}
{\color{black} where $t_i^{\ell,h_\ell}$ is the index of the first nonzero column of $C_i^{\ell,h_\ell}$.} 
Clearly, $\mathbb{G}_{i,1}^{\ell} \cup \mathbb{G}_{i,2}^{\ell} \cup \mathbb{G}_{i,3}^{\ell} = [1, g_\ell].$ On the other hand, for every pair $(\ell, h_\ell) \in [1,r_u] \times [1, g_\ell]$, 
we can  introduce the index sets
\begin{equation}\label{eq:defVlhl}
\mathcal{V}_{k}^{\ell,h_\ell}\doteq\{i \in [1,N]: h_\ell\in\mathbb{G}_{i,1}^{\ell}\}, \quad k\in [1,3]. 
\end{equation}
Note that the three sets are disjoint, and
$\mathcal{V}_{1}^{\ell,h_\ell} \cup \mathcal{V}_{2}^{\ell,h_\ell} \cup \mathcal{V}_{3}^{\ell,h_\ell} =[1,N].$ We want to comment on the meaning of the previous sets. It is well-known (and an immediate consequence of the PBH observability test \cite{Kailath}) that a pair $(A^{\ell,h_\ell},C_i^{\ell, h_\ell})$, where $A^{\ell,h_\ell}$ is a Jordan miniblock associated with the eigenvalue $\lambda_\ell$, is observable if and only if the first column of $C_i^{\ell, h_\ell}$ is nonzero (i.e., $t_i^{\ell, h_\ell} =1$). Under this perspective, the sets $\mathbb{G}_{i,k}^{\ell}, k\in[1,3],$ defined in \eqref{eq:defGil} identify, for a given agent $i$ and a certain eigenvalue $\lambda_\ell, \ell\in [1,r_u]$, which Jordan miniblocks $A^{\ell,h_\ell}$ are completely unobservable, partly observable and completely observable. The second family of sets, $\mathcal{V}_{k}^{\ell,h_\ell}, k\in[1,3],$ defined in \eqref{eq:defVlhl}, on the other hand, identifies for each Jordan miniblock $A^{\ell,h_\ell}$ the agents for which the block is completely unobservable, partly observable and completely observable.

The following proposition shows that, under Assumption \ref{ass.jointdet}, for each Jordan miniblock  $A^{\ell,h_\ell}$ there exists at least one agent $i\in [1,N]$ such that the pair $(A^{\ell,h_\ell}, C_i^{\ell,h_\ell})$ is completely observable.
\smallskip

\begin{proposition}
Under Assumption \ref{ass.jointdet}, for every pair $(\ell, h_\ell)\in [1,r_u] \times [1, g_\ell]$, the set ${\mathcal V}_3^{\ell,h_\ell}$ is not empty, or, equivalently,
there exists $i\in [1,N]$ such that $h_\ell\in {\mathbb G}_{i,3}^\ell.$
\end{proposition}
\begin{proof}
    By Assumption \ref{ass.jointdet} and the hypothesis that $A$ is  Jordan form, it follows from the PBH observability test that for every $(\ell, h_\ell)\in [1,r_u] \times [1, g_\ell]$ there exists $i\in [1,N]$ such that $(C_i^{\ell,h_\ell}\ne\mathbb{0}) \wedge\ (t_i^{\ell, h_\ell} = 1).$
    This implies that $h_\ell \in {\mathbb G}_{i,3}^\ell$.
\end{proof}

Based on the definitions \eqref{eq:defGil} and the fact that $\mathbb{G}_{i,1}^{\ell} \cup \mathbb{G}_{i,2}^{\ell} \cup \mathbb{G}_{i,3}^{\ell} = [1, g_\ell]$, for every $i\in[1,N]$ and $\ell\in [1,r_u]$
there exists a permutation matrix $P^{\ell}_i \in {\mathbb R}^{a_\ell \times a_\ell}$  such that 
\begin{align*}
 ({P^{\ell}_i})^\top A^{\ell}P^{\ell}_i&=\begin{bmatrix}
    A_{i,1}^{\ell} &\vline&\mathbb{0}&\mathbb{0}\\
       \hline
    \mathbb{0}&\vline&A_{i,2}^{\ell} &\mathbb{0}\\
    \mathbb{0}&\vline&\mathbb{0}&A_{i,3}^{\ell}
\end{bmatrix},
  ({P^{\ell}_i})^\top B^\ell=\begin{bmatrix}
    B_{i,1}^{\ell}\\
    \hline
    B_{i,2}^{\ell}\\
    B_{i,3}^{\ell}
\end{bmatrix}\\
C_iP^{\ell}_i&=\begin{bmatrix}
    \mathbb{0}&\vline&C_{i,2}^{\ell}&C_{i,3}^{\ell}
\end{bmatrix},
 ({P^{\ell}_i})^\top x^\ell(t)=\begin{bmatrix}
    x_{i,1}^{\ell}(t)\\
    \hline
    x_{i,2}^{\ell}(t)\\
    x_{i,3}^{\ell}(t)
\end{bmatrix}, 
\end{align*}
where the block diagonal matrices $A_{i,1}^\ell, A_{i,2}^\ell$ and $A_{i,3}^\ell$ group all Jordan miniblocks corresponding to indices $h_\ell$ belonging to ${\mathbb G}_{i,1}^{\ell}, {\mathbb G}_{i,2}^{\ell}$ and ${\mathbb G}_{i,3}^{\ell}$, respectively, i.e.,
\begin{equation*}
     A_{i,k}^\ell \doteq{\rm diag}\{A^{\ell,h_\ell}\}_{h_\ell\in\mathbb{G}_{i,k}^{\ell}},\quad
    k\in [1,3].
\end{equation*}
The blocks of $(P_i^\ell)^\top B^{\ell}$ and $C_iP_i^\ell$ are defined accordingly.
In particular,  the matrix $C_{i,2}^\ell$ (respectively, $C_{i,3}^\ell$)  consists of the blocks $C_i^{\ell,h_\ell}$ of the matrix $C_i$ corresponding to the Jordan miniblocks $A^{\ell,h_\ell}$ with $h_\ell\in \mathbb{G}_{i,2}^{\ell}$ (respectively, $h_\ell\in \mathbb{G}_{i,3}^{\ell}$).
\smallskip

Note that if $h_\ell \in \mathbb{G}_{i,2}^{\ell}$,
then $C_i^{\ell,h_\ell}$ takes the form
$$C_i^{\ell,h_\ell} = \begin{bmatrix} {\mathbb 0}_{p_i\times (t_i^{\ell,h_\ell}-1)} & C_{i,o}^{\ell,h_\ell}\end{bmatrix},$$
and 
the first column of $C_{i,o}^{\ell,h_\ell}$ is nonzero.  Moreover, the corresponding  
blocks $A^{\ell, h_\ell}$  and  {\color{black}$x^{\ell,h_\ell}(t)$} can be  
partitioned as follows  
\begin{equation}A^{\ell,h_\ell} = \begin{bmatrix} A^{\ell,h_\ell}_{i,u} &A^{\ell,h_\ell}_{i,*}\cr\mathbb{0} & A^{\ell,h_\ell}_{i,o}\end{bmatrix}, \quad
{\color{black}x^{\ell,h_\ell}(t) = \begin{bmatrix} x_{i,u}^{\ell, h_\ell}(t) \cr x_{i,o}^{\ell, h_\ell}(t)
\end{bmatrix}},
\label{eq.fatt.utile}
\end{equation}
where $A^{\ell,h_\ell}_{i,u}$ and $A^{\ell,h_\ell}_{i,o}$ 
are described as in \eqref{Jminiblock} (and hence 
have, in turn, the structure of two Jordan miniblocks) 
and have  sizes $t_i^{\ell,h_\ell}-1$ and $d^{\ell, h_\ell} - (t_i^{\ell,h_\ell}-1)$, respectively. 
It is easy to see that the pair $(A^{\ell,h_\ell}_{i,o},C_{i,o}^{\ell,h_\ell})$ is observable.  
\medskip

{\bf Step 2: Grouping within $A_{i,2}^\ell$ of the observable and unobservable parts of all miniblocks $A^{\ell, h_\ell}$, $h_\ell\in {\mathbb G}_{i,2}^{\ell, h_\ell}$}. \\  We now focus on the pair $(A_{i,2}^\ell,C_{i,2}^\ell)$ and permute the entries in such a way that  the entries corresponding to all the blocks $A_{i,u}^{\ell, h_\ell}$  come first, while the entries corresponding to all the blocks $A_{i,o}^{\ell, h_\ell}$ appear at the end. This amounts to considering the permuted matrices
\begin{align*}
{(P_{i,2}^\ell)}^\top A_{i,2}^\ell{P_{i,2}^\ell}&=\begin{bmatrix}
A_{i,2u}^{\ell}&\vline&A_{i,2*}^{\ell}\\
    \hline
 \mathbb{0} &\vline& A_{i,2o}^{\ell}
\end{bmatrix}, \ {(P_{i,2}^\ell)}^\top  B^\ell=\begin{bmatrix}
    B_{i,2u}^{\ell}\\
    \hline
    B_{i,2o}^{\ell}
\end{bmatrix},\\
C_{i,2}^\ell{P_{i,2}^\ell}&=\begin{bmatrix}
    \mathbb{0}&\vline&C_{i,2o}^{\ell}
\end{bmatrix}, \quad {(P_{i,2}^\ell)}^\top x_{i,2}^\ell=\begin{bmatrix}
    x_{i,2u}^{\ell}\\
    \hline
    x_{i,2o}^{\ell}
\end{bmatrix},
\end{align*}
where $A_{i,2u}^\ell, A_{i,2o}^\ell$ and $C_{i,2o}^{\ell}$ consist of all blocks $A_{i,u}^{\ell,h_\ell}, A_{i,o}^{\ell,h_\ell}$ and
$C_{i,o}^{\ell,h_\ell}$, respectively,  corresponding to the indices $h_\ell\in \mathbb{G}_{i,2}^{\ell}$, namely
\begin{align*}
A_{i,2u}^\ell &\doteq {\rm diag} \{A^{\ell,h_\ell}_{i,u} \}_{h_\ell \in {\mathbb G}_{i,2}^\ell},\
A_{i,2o}^\ell \doteq {\rm diag} \{A^{\ell,h_\ell}_{i,o} \}_{h_\ell \in {\mathbb G}_{i,2}^\ell}, \\
C_{i,2o}^{\ell}&\doteq {\rm row} \{C^{\ell,h_\ell}_{i,o} \}_{h_\ell \in {\mathbb G}_{i,2}^\ell}.
\end{align*}
The other blocks are defined accordingly.
\smallskip

{\bf Step 3: Grouping within $A$ of blocks with the same observability properties for agent $i$ but corresponding to distinct eigenvalues.}\\ We now combine the outcome of the two previous steps and introduce a final permutation that groups  blocks with the same observability properties (for agent $i$) corresponding to distinct (unstable) eigenvalues. We define  
\begin{align*}
 A_{i,k}&\doteq{\rm diag}\{A_{i,k}^{\ell}\}_{\ell\in [1,r_u]}, \quad B_{i,k}\doteq{\rm col}\{{B_{i,k}^{\ell}}\}_{\ell\in [1,r_u]} \\
 x_{i,k}(t)&\doteq{\rm col}\{{x_{i,k}^{\ell}(t)}\}_{\ell\in [1,r_u]}, \quad \quad k=1,3,
  \\
A_{i,2u}&\doteq{\rm diag}\{A_{i,2u}^{\ell}\}_{\ell\in [1,r_u]},\quad  
A_{i,2*}\doteq{\rm diag}\{A_{i,2*}^{\ell}\}_{\ell\in [1,r_u]},\\ 
A_{i,2o}&\doteq{\rm diag}\{A_{i,2o}^{\ell}\}_{\ell\in [1,r_u]},
\\
 B_{i,2u}&\doteq{\rm col}\{{B_{i,2u}^{\ell}}\}_{\ell\in [1,r_u]},\quad 
B_{i,2o}\doteq{\rm col}\{{B_{i,2o}^{\ell}}\}_{\ell\in [1,r_u]},\\
C_{i,2o}&\doteq{\rm row}\{{C_{i,2o}^{\ell}}\}_{\ell\in [1,r_u]},\quad
C_{i,3}\doteq{\rm row}\{{C_{i,3}^{\ell}}\}_{\ell\in [1,r_u]}, \\
x_{i,2u}(t)&\doteq{\rm col}\{{x_{i,2u}^{\ell}(t)}\}_{\ell\in [1,r_u]},
x_{i,2o}(t)\doteq{\rm col}\{{x_{i,2o}^{\ell}(t)}\}_{\ell\in [1,r_u]}
\end{align*}
We also group together Jordan blocks corresponding to stable eigenvalues:
\begin{align*}
A_{s}&\doteq{\rm diag}\{A^{\ell}\}_{\ell\in[r_u+1,r]},\quad B_{s}\doteq{\rm diag}\{B^{\ell}\}_{\ell\in[r_u+1,r]},\\
 C_{s}& \doteq{\rm row}\{C^{\ell}\}_{\ell\in[r_u+1,r]}, \quad x_{s}(t)\doteq{\rm col}\{{x^{\ell}(t)}\}_{\ell\in [r_u+1,r]}.
\end{align*}
As a result of all previous re-orderings and definitions,  we can claim that for every $i\in [1,N]$
there exists a permutation matrix $Q_i\in {\mathbb R}^{n \times n}$ such that 
\begin{subequations}\label{eq.ste3.perm}
\begin{align}
{Q_i}^\top AQ_i&=\begin{bmatrix}
    F_{i,u}&\vline& F_{i,*}\\
    \hline
    \mathbb{0}&\vline& F_{i,d}
\end{bmatrix}\\
&\doteq %\scriptsize
 \begin{bmatrix}
    A_{i,1}&\mathbb{0} &\vline&\mathbb{0}&\mathbb{0}&\mathbb{0}\\
    \mathbb{0}&A_{i,2u}&\vline&A_{i,2*} &\mathbb{0}&\mathbb{0}\\
    \hline
    \mathbb{0} &\mathbb{0}&\vline&A_{i,2o}&\mathbb{0}&\mathbb{0} \\
    \mathbb{0}&  \mathbb{0}&\vline&\mathbb{0}&A_{i,3}&\mathbb{0}\\
        \mathbb{0}&  \mathbb{0}&\vline&\mathbb{0}&\mathbb{0}&A_s
\end{bmatrix} \nonumber\\
C_iQ_i&=\begin{bmatrix}
    \mathbb{0}&\vline&H_{i,d}
\end{bmatrix} 
\doteq\begin{bmatrix}
    \mathbb{0}&\mathbb{0}&\vline&C_{i,2o}&C_{i,3}&C_{s}
\end{bmatrix}, \\
{Q_i}^\top x(t)&=\begin{bmatrix}
    z_{i,u}(t)\\
    \hline
    z_{i,d}(t)
\end{bmatrix}
\doteq %{\color{black}\scriptsize
\begin{bmatrix}
    x_{i,1}(t)\\
    x_{i,2u}(t)\\
    \hline
    x_{i,2o}(t)\\
    x_{i,3}(t)\\
    x_{s}(t)
\end{bmatrix}, \label{eq.zdzu}\\ {Q_i}^\top B&=\begin{bmatrix}
    G_{i,u}\\
    \hline
    G_{i,d}
\end{bmatrix} 
\doteq %{\color{black}\scriptsize
\begin{bmatrix}
    B_{i,1}\\
    B_{i,2u}\\
    \hline
    B_{i,2o}\\
    B_{i,3}\\
    B_{s}
\end{bmatrix}.  
\end{align}
\end{subequations}
From now on, we will rely on the following assumption\footnote{Note that this assumption is also used in \cite{GaoYang}, but it remains hidden. Without this assumption, the matrix governing the estimation error dynamics of the Luenberger observer proposed in (21) of  \cite{GaoYang} would not be Schur stable.}.
\smallskip

\begin{assumption}\label{ass.obs}
    For every $i\in[1,N]$ and $\ell\in [1,r_u]$, the   vectors
    $
    \{C_i^{\ell,h_\ell}\mathbb{e}_{t_i^{\ell,h_\ell}}
    \}_{h_\ell\in\mathbb{G}_{i,2}^{\ell}\cup  \mathbb{G}_{i,3}^{\ell}} = \{C_{i,o}^{\ell,h_\ell}\mathbb{e}_{1}
    \}_{h_\ell\in\mathbb{G}_{i,2}^{\ell}} \cup \{C_i^{\ell,h_\ell}\mathbb{e}_{1}\}_{h_\ell\in\ \mathbb{G}_{i,3}^{\ell}}
     $
    are linearly independent.  
    \end{assumption}

We remark that if Assumption \ref{ass.obs} holds, then it follows from the observability of the pairs 
$(A^{\ell,h_\ell},C_i^{\ell,h_\ell})$  for all $h_\ell\in\mathbb{G}_{i,3}^{\ell}$, and of the pairs 
$(A_{i,o}^{\ell,h_\ell},C_{i,o}^{\ell,h_\ell})$ for all $h_\ell\in \mathbb{G}_{i,2}^{\ell}$, 
 that
for every $i\in[1,N]$ and $\ell\in [1,r_u]$, the pair 
$$\left(
    \begin{bmatrix}
        A_{i,2o}^{\ell}&\mathbb{0}\\
        \mathbb{0}&A_{i,3}^{\ell}
    \end{bmatrix},\begin{bmatrix}
    C_{i,2o}^{\ell}& C_{i,3}^{\ell}
\end{bmatrix}\right)$$ 
    is observable. Therefore for every $i\in [1,N]$ the pair $(F_{i,d}, H_{i,d})$ (that includes also the stable eigenvalues of $A$, namely the block $A_s$, and the corresponding blocks in $C_i$) is detectable.
    \\
    The preceding discussion clarifies the meaning of Assumption \ref{ass.obs}. It requires that, for each pair $(A, C_i), i\in [1,N],$ there exists a permutation matrix 
$Q_i$  such that the transformed pair is in Kalman detectability form.

\section{First strategy}\label{sec.strategy1}

We now exploit the fact that  each pair $(A,C_i)$ has been reduced to Kalman detectability form
to provide a distributed estimation scheme in which   each agent $i\in [1,N]$ separately estimates the portions of the  state that are detectable and   undetectable for it, namely $z_{i,d}(t)$ and $z_{i,u}(t)$, respectively.

{\bf Estimate of $z_{i,d}(t)$:}\ By relying on the fact that for each $i\in [1,N]$ the pair 
$(F_{i,d}, H_{i,d})$ is detectable, we can
 design a Luenberger observer that allows agent $i$ to estimate the part of the vector $x(t)$ that is detectable for it, namely $z_{i,d}(t)$ (see \eqref{eq.zdzu}), as follows:
\begin{align}\label{eq.luenberger}
\hat z_{i,d}(t+1)=&F_{i,d}\hat z_{i,d}(t)+G_{i,d}u(t)\\
    &+L_{i,d}(y_i(t)-H_{i,d}\hat z_{i,d}(t)),\nonumber
\end{align}
where $L_{id}$ is chosen so that $F_{i,d} - L_{i,d} H_{i,d}$ is Schur stable. 
If we define the estimation error of the part of the state that is detectable for agent $i$ as $\eta_{i,d}(t) \doteq z_{i,d}(t)-\hat z_{i,d}(t)$, then the estimation error dynamics are given by
\begin{align}
     \eta_{i,d}(t+1)=(F_{i,d}-L_{i,d}H_{i,d}) \eta_{i,d}(t),
\end{align} 
and $\eta_{i,d}(t)$ asymptotically goes to zero, for every 
$\eta_{i,d}(0)$. 
\smallskip

{\bf Estimate of $z_{i,u}(t)$:}\ In order to estimate the remaining part of the state, denoted by $z_{i,u}(t)$, 
agent $i$ has to rely on the information it can acquire from its neighbors, since the block of  $C_i Q_i$ corresponding to this portion of the state vector is the zero matrix. Specifically, 
 the estimate  $\hat z_{i,u}(t)$ is  obtained using the estimates   of the state entries appearing in $z_{i,u}(t)$ provided by each agent $j
 \in {\mathcal N}_i$. 
 {\color{black} Such estimates are obtained} in this way: first,  (by reversing \eqref{eq.zdzu}) we {\color{black} express  the estimate $\hat x_j(t)$ of the entire state $x(t)$ provided by agent $j$ 
  as}
\begin{align*}
\hat x_j(t) =  {Q_j} 
\begin{bmatrix}
    \hat z_{j,u}(t)\\
    \hline
    \hat z_{j,d}(t)
\end{bmatrix}.
\end{align*}
Secondly, the estimate of $z_{i,u}(t)$ provided by  agent $j\in {\mathcal N}_i$ can be obtained as
$\begin{bmatrix}
        I&\mathbb{0}
    \end{bmatrix}Q_i^\top \hat x_j(t),$
    where the size of the identity matrix coincides with the {\color{black} size} of $z_{i,u}(t)$.
    \\
In detail, we implement the following observer for $z_{i,u}(t)$:
\begin{align}\label{eq.estimate.undetectble.reduced}
    \hat z_{i,u}(t+1)=&F_{i,u}\hat z_{i,u}(t)+F_{i,*}\hat z_{i,d}(t)+G_{i,u} u(t)\\
                    &+K_iF_{i,u}\sum_{j\in\mathcal{N}_i}[\mathcal{A}]_{i,j}\begin{bmatrix}
        I&\mathbb{0}
    \end{bmatrix}Q_i^\top(\hat x_j(t)-\hat x_i(t))\nonumber
\end{align}
where 
the matrix ${\mathcal A}$ is the adjacency matrix of the communication graph ${\mathcal G}$, and $K_i$ is a block-diagonal matrix of the following form:
 \begin{align*}
  K_i \doteq \begin{bmatrix}{\rm diag} \left\{ K_{i,1}^\ell\right\}_{\ell\in [1,r_u]} & \mathbb{0} \cr \mathbb{0}&{\rm diag} \left\{K_{i,2}^\ell\right\}_{\ell\in [1,r_u]  }\end{bmatrix},
  \end{align*}
where
\begin{align}
K_{i,1}^\ell \doteq&\ {\rm diag} \left\{ k^{\ell,h_\ell} I_{d^{\ell, h_\ell}}\right\}_{ h_\ell \in {\mathbb G}_{i,1}^\ell},
\label{Kappa1}\\
K_{i,2}^\ell \doteq&\ {\rm diag} \left\{ k^{\ell,h_\ell} I_{t_i^{\ell, h_\ell}-1}\right\}_{ h_\ell \in {\mathbb G}_{i,2}^\ell}, \label{Kappa2}
\end{align}
and $k^{\ell, h_\ell}$ is a real design parameter.
 It remains to determine under what conditions  for each pair $(\ell, h_\ell)$ such that   ${\mathcal V}_3^{\ell, h_\ell} \subsetneq [1,N]$ (i.e., the block  $x^{\ell,h_\ell}(t)$ is not completely observable for all agents),  there exists\footnote{ \label{footn:numerok}
    The number of  gains $k^{\ell,h_\ell}$ to be designed  coincides with the cardinality of the set 
    ${\mathcal U}_\ell \doteq \{ (\ell, h_\ell): (\ell \in [1,r_u]) \wedge ({\mathcal V}_3^{\ell, h_\ell} \ne [1,N])\} \subseteq [1,r_u]\times [1,g_\ell],$
 which  is upper bounded by 
    $\sum_{\ell=1}^{r_u}g_\ell$, the number of miniblocks of $A$ associated to unstable eigenvalues. In the sequel, we will   write $(\ell, h_\ell) \in [1,r_u]\times [1,g_\ell]$ rather than specifying that $(\ell, h_\ell) \in {\mathcal U}_\ell.$} $k^{\ell, h_\ell}\in {\mathbb R}$ such that the estimate $\hat z_{i,u}(t)$ of the undetectable part of the state of each agent $i\in {\mathcal V}_1^{\ell, h_\ell} \cup {\mathcal V}_2^{\ell, h_\ell}$ asymptotically converges to the true value.    To this end, for each such node $i$, we   define  the estimation error of the undetectable part as $\eta_{i,u}(t) \doteq z_{i,u}(t)-\hat z_{i,u}(t)$, whose dynamics are equal to 
\begin{align*}
    \eta_{i,u}(t+1)=&F_{i,u}\eta_{i,u}(t)+F_{i,*} \eta_{i,d}(t)\\
            &-K_iF_{i,u}\sum_{j\in\mathcal{N}_i}[\mathcal{A}]_{i,j}\begin{bmatrix}
        I&\mathbb{0}
    \end{bmatrix}Q_i^\top(\hat x_j(t)-\hat x_i(t))\\
        =&(I - K_i[\mathcal{L}]_{i,i})F_{i,u}\eta_{i,u}(t)+F_{i,*} \eta_{i,d}(t)\\
            &- K_iF_{i,u}\sum_{j\in\mathcal{N}_i}[\mathcal{L}]_{i,j}\begin{bmatrix}
        I&\mathbb{0}
    \end{bmatrix}Q_i^\top Q_j \begin{bmatrix}\eta_{j,u}(t)\cr \eta_{j,d}(t)\end{bmatrix},
\end{align*}
{\color{black}     where we have used the definition of the Laplacian matrix $\mathcal{L}$.}
As a result, the overall dynamics of the observer estimation error can be written as 
\begin{equation}
\label{eq.triang.und}
\begin{bmatrix} \eta_u(t+1)\cr \eta_d(t+1)\end{bmatrix}
= \begin{bmatrix} \Gamma_u & \Gamma_* \cr {\mathbb 0} & \Gamma_d\end{bmatrix}
\begin{bmatrix} \eta_u(t)\cr \eta_d(t)\end{bmatrix}
\end{equation}
where  $\eta_d(t) \doteq {\rm col}\{\eta_{i,d}(t)\}_{i\in [1,N]}$, $\eta_u(t) \doteq {\rm col}\{\eta_{i,u}(t)\}_{i\in [1,N]}$,  $\Gamma_d \doteq {\rm diag} \{ F_{i,d} - L_{i,d}H_{i,d}\}_{i\in [1,N]}$ is Schur, while $\Gamma_u$ and $
\Gamma_*$ are  suitable matrices that can be obtained from the previous equations describing the dynamics of the vectors $\eta_{i,u}(t)$.
While the {\color{black} expression of $\Gamma_*$ is irrelevant for the convergence analysis, that}  of $\Gamma_u$ will be derived later. 
The convergence to zero of $\eta_{d}(t)$   has already been discussed. On the other hand, due to the block triangular structure of the matrix in \eqref{eq.triang.und}, 
condition 
 \begin{align}
     \lim_{t\to+\infty}\eta_{u}(t)=\mathbb{0},
 \end{align}
holds 
 if and only if 
the matrix $\Gamma_u$ is Schur.

To determine the conditions under which we can find  parameters $k^{\ell,h_\ell}$ that guarantee the convergence to zero of the estimation error $\eta_u(t)$, we consider the dynamics of each block that appears in $z_{i,u}(t)$. This will also lead us to determine the expression of $\Gamma_u$.
 For every $i\in {\mathcal V}_1^{\ell, h_\ell}\cup {\mathcal V}_2^{\ell, h_\ell}$, introduce the matrix with $d^{\ell, h_\ell}$ columns
\begin{align*}
S_i^{\ell, h_\ell} \doteq \begin{cases}  I_{d^{\ell,h_\ell}}, & {\rm if} \ i\in {\mathcal V}_1^{\ell, h_\ell};\cr
\begin{bmatrix} I_{t_i^{\ell,h_\ell}-1} & {\mathbb 0}\end{bmatrix},& {\rm if} \ i\in {\mathcal V}_2^{\ell, h_\ell},
\end{cases}
\end{align*}
and
we distinguish two cases: $i\in {\mathcal V}_1^{\ell, h_\ell}$ and $i\in {\mathcal V}_2^{\ell, h_\ell}$.
\\
    $\bullet$   \underline{For $i\in {\mathcal V}_1^{\ell, h_\ell}$}, the estimate of  $x^{\ell, h_\ell}(t)$ updates as
\begin{align}
\label{eq.estimate.undetectble.directed}
        \hat x_{i}^{\ell,h_\ell}&(t+1)=
        A^{\ell,h_\ell}
        \hat x_{i}^{\ell,h_\ell}(t)
+B^{\ell,h_\ell} u(t)\\&+k^{\ell,h_\ell}
    A^{\ell,h_\ell}\sum_{j\in\mathcal{N}_i}[\mathcal{A}]_{i,j}(\hat x_j^{\ell,h_\ell}(t)-\hat x_i^{\ell,h_\ell}(t)).\nonumber
\end{align}
Therefore its estimation error $e_{i}^{\ell,h_\ell}
(t) \doteq {\color{black}x^{\ell, h_\ell}(t) - \hat x_i^{\ell, h_\ell}(t)}$ updates as:
        \begin{align} 
        e_{i}^{\ell,h_\ell}(t+1)=&(1 -  k^{\ell,h_\ell}[\mathcal{L}]_{i,i})
         A^{\ell,h_\ell}e_{i}^{\ell,h_\ell}(t) \nonumber \\
         &- k^{\ell,h_\ell}
        A^{\ell,h_\ell}\sum_{j\in \mathcal{N}_{i}}[\mathcal{L}]_{i,j}e_j^{\ell,h_\ell}(t) \nonumber\\
        =&(1 -  k^{\ell,h_\ell}[\mathcal{L}]_{i,i})
         A^{\ell,h_\ell}e_{i}^{\ell,h_\ell}(t) \label{eilh}\\
         &- k^{\ell,h_\ell}
        A^{\ell,h_\ell}\sum_{j\in  {\mathcal N}_i \cap\mathcal{V}_{1}^{\ell,h_\ell}}[\mathcal{L}]_{i,j}e_j^{\ell,h_\ell}(t) \nonumber
        \\&- k^{\ell,h_\ell}
        A^{\ell,h_\ell}\sum_{j\in {\mathcal N}_i \cap\mathcal{V}_{2}^{\ell,h_\ell}}[\mathcal{L}]_{i,j}e_j^{\ell,h_\ell}(t) \nonumber
        \\
        &- k^{\ell,h_\ell}
        A^{\ell,h_\ell}\sum_{j\in {\mathcal N}_i \cap \mathcal{V}_{3}^{\ell,h_\ell}}[\mathcal{L}]_{i,j}e_j^{\ell,h_\ell}(t). \nonumber
    \end{align}
By using \eqref{eq.fatt.utile}, the previous equation becomes:
\begin{align*} 
       \!\!\!\! e_{i}^{\ell,h_\ell}(t+1)&=(1 -  k^{\ell,h_\ell}[\mathcal{L}]_{i,i})
         A^{\ell,h_\ell}e_{i}^{\ell,h_\ell}(t)\\
         &- k^{\ell,h_\ell}
        A^{\ell,h_\ell}\sum_{j\in {\mathcal N}_i \cap \mathcal{V}_{1}^{\ell,h_\ell}}[\mathcal{L}]_{i,j}e_j^{\ell,h_\ell}(t) \nonumber
        \\&- k^{\ell,h_\ell}\!\!\!\!\!\!
        \sum_{j\in {\mathcal N}_i \cap\mathcal{V}_{2}^{\ell,h_\ell}}\!\!\!\!\!\![\mathcal{L}]_{i,j}\begin{bmatrix} A^{\ell,h_\ell}_{j,u} &A^{\ell,h_\ell}_{j,*}\cr\mathbb{0} & A^{\ell,h_\ell}_{j,o}\end{bmatrix} \!\!\!
        {\color{black}\begin{bmatrix}e_{j,u}^{\ell,h_\ell}(t)\\[0.5ex]e_{j,o}^{\ell,h_\ell}(t)
        \end{bmatrix}} \nonumber
        \\&- k^{\ell,h_\ell}
        A^{\ell,h_\ell}\sum_{j\in {\mathcal N}_i \cap\mathcal{V}_{3}^{\ell,h_\ell}}[\mathcal{L}]_{i,j}e_j^{\ell,h_\ell}(t),\nonumber
\end{align*}
where {\color{black}$e_{j,u}^{\ell,h_\ell}(t) \doteq x_{j,u}^{\ell,h_\ell}(t)-\hat x_{j,u}^{\ell,h_\ell}(t)$ and $e_{j,o}^{\ell,h_\ell}(t) \doteq x_{j,o}^{\ell,h_\ell}(t)-\hat x_{j,o}^{\ell,h_\ell}(t)$}. 
To simplify the notation and highlight the information that is related only to the detectable parts of the state estimates, we can rewrite the previous equation as follows:
\begin{align}\label{eilh2}
        e_{i}^{\ell,h_\ell}(t+1)&=(1 -  k^{\ell,h_\ell}[\mathcal{L}]_{i,i})
         A^{\ell,h_\ell}e_{i}^{\ell,h_\ell}(t)\\
         &- k^{\ell,h_\ell}
        A^{\ell,h_\ell}\sum_{j\in {\mathcal N}_i \cap\mathcal{V}_{1}^{\ell,h_\ell}}[\mathcal{L}]_{i,j}e_j^{\ell,h_\ell}(t) \nonumber
        \\&- k^{\ell,h_\ell}
        \sum_{j\in  {\mathcal N}_i \cap\mathcal{V}_{2}^{\ell,h_\ell}}[\mathcal{L}]_{i,j} A^{\ell, h_\ell} (S_j^{\ell, h_\ell})^\top {\color{black}e_{j,u}^{\ell,h_\ell}(t)}\nonumber
        \\&+  \Gamma_{i,*}^{\ell,h_\ell}\eta_d(t),\nonumber
\end{align}
where {\color{black}$A^{\ell, h_\ell} (S_j^{\ell, h_\ell})^\top
    =    [(A^{\ell,h_\ell}_{j,u})^\top \ \mathbb{0}]^\top$}, 
 while the term $\Gamma_{i,*}^{\ell,h_\ell}\eta_d(t)$ 
 groups together the terms  $-k^{\ell,h_\ell}
        [\mathcal{L}]_{i,j}$ ${\color{black}[(A^{\ell,h_\ell}_{j,*})^\top\ (A^{\ell,h_\ell}_{j,o})^\top]^\top}e_{j,2o}^{\ell,h_\ell}(t)$ corresponding to the indices $j\in 
        \mathcal{V}_{2}^{\ell,h_\ell}$, and the terms
        $- k^{\ell,h_\ell}
        A^{\ell,h_\ell}[\mathcal{L}]_{i,j}e_j^{\ell,h_\ell}(t)$ corresponding to the indices $j\in \mathcal{V}_{3}^{\ell,h_\ell}.$

$\bullet$  \underline{For  $i\in {\mathcal V}_2^{\ell, h_\ell}$}, the estimate of ${\color{black}x_{i,u}^{\ell, h_\ell}(t)} = S_i^{\ell, h_\ell} x^{\ell, h_\ell}(t)$, the upper part of
$x^{\ell, h_\ell}(t)$, updates as:
      {\color{black} \begin{align*}
 \hat x_{i,u}^{\ell,h_\ell}&(t+1)=
        A_{i,u}^{\ell,h_\ell}\
        \hat x_{i,u}^{\ell,h_\ell}(t)
        +A_{i,*}^{\ell,h_\ell}\
        \hat x_{i,o}^{\ell,h_\ell}(t)
+B_{i,u}^{\ell,h_\ell} u(t)\\&+k^{\ell,h_\ell}
    \begin{bmatrix}A_{i,u}^{\ell,h_\ell} &A_{i,*}^{\ell,h_\ell} \end{bmatrix} \sum_{j\in\mathcal{N}_i}[\mathcal{A}]_{i,j}(\hat x_j^{\ell,h_\ell}(t)-\hat x_i^{\ell,h_\ell}(t)).\nonumber
\end{align*}}
Thus the corresponding estimation error updates as
        {\color{black}\begin{align*}%\label{eilh}
        e_{i,u}^{\ell,h_\ell}&(t+1)=(1 -  k^{\ell,h_\ell}[\mathcal{L}]_{i,i})
         A^{\ell,h_\ell}_{i,u}\ e_{i,u}^{\ell,h_\ell}(t) \nonumber
         \\&+  (1 -  k^{\ell,h_\ell}[\mathcal{L}]_{i,i}) A_{i,*}^{\ell,h_\ell}\ e_{i,o}^{\ell,h_\ell}(t) \nonumber
         \\&-k^{\ell,h_\ell}
    \begin{bmatrix}A_{i,u}^{\ell,h_\ell} &A_{i,*}^{\ell,h_\ell} \end{bmatrix} \sum_{j\in\mathcal{N}_i}[\mathcal{L}]_{i,j}\ e_j^{\ell,h_\ell}(t) \nonumber\\
    &=(1 -  k^{\ell,h_\ell}[\mathcal{L}]_{i,i})
         A^{\ell,h_\ell}_{i,u}\ e_{i,u}^{\ell,h_\ell}(t)\\
         \\&-k^{\ell,h_\ell}
    \begin{bmatrix}A_{i,u}^{\ell,h_\ell} &A_{i,*}^{\ell,h_\ell} \end{bmatrix} \sum_{j\in {\mathcal N}_i \cap\mathcal{V}_1^{\ell, h_\ell}}[\mathcal{L}]_{i,j}\ e_j^{\ell,h_\ell}(t) \nonumber
    \\&-k^{\ell,h_\ell}
    \sum_{j\in {\mathcal N}_i \cap\mathcal{V}_2^{\ell, h_\ell}} [\mathcal{L}]_{i,j}\ \begin{bmatrix}A_{i,u}^{\ell,h_\ell} &A_{i,*}^{\ell,h_\ell} \end{bmatrix} (S_j^{\ell,h_\ell})^\top
    e_{j,u}^{\ell,h_\ell}(t) \nonumber
    \\& +  \Gamma_{i,*}^{\ell,h_\ell} \eta_d(t)
 \nonumber\\
 &=(1 -  k^{\ell,h_\ell}[\mathcal{L}]_{i,i})
         S_i^{\ell, h_\ell} A^{\ell,h_\ell} (S_i^{\ell, h_\ell} )^\top \ e_{i,u}^{\ell,h_\ell}(t)\\
         \\&-k^{\ell,h_\ell}
    S_i^{\ell, h_\ell} A^{\ell,h_\ell}  \sum_{j\in {\mathcal N}_i \cap\mathcal{V}_1^{\ell, h_\ell}}[\mathcal{L}]_{i,j}\ e_j^{\ell,h_\ell}(t) \nonumber
    \\&-k^{\ell,h_\ell}
    \sum_{j\in {\mathcal N}_i \cap\mathcal{V}_2^{\ell, h_\ell}} [\mathcal{L}]_{i,j}\ S_i^{\ell, h_\ell} A^{\ell,h_\ell}  (S_j^{\ell,h_\ell})^\top
    e_{j,u}^{\ell,h_\ell}(t) \nonumber
    \\& + \Gamma_{i,*}^{\ell,h_\ell} \eta_d(t),
 \nonumber
\end{align*}}
where, again, $\Gamma_{i,*}^{\ell,h_\ell} \eta_d(t)$ keeps into account the terms that belong to the detectable parts of the state vectors both for agent $i$ and for its neighbors, i.e., those depending on $e_{i,o}^{\ell, h_\ell}, e_{j,o}^{\ell, h_\ell}, j\in {\mathcal V}_2^{\ell, h_\ell}\cap {\mathcal N}_i,$ and $e_j^{\ell, h_\ell}(t), j\in {\mathcal V}_3^{\ell,h_\ell}$. 

 If we denote by $\eta_u^{\ell,h_\ell}(t)$  the portion of the estimation error 
$\eta_u(t)$ that corresponds to the Jordan miniblock $A^{\ell, h_\ell}$ and we assume that $${\mathcal V}_1^{\ell, h_\ell} \cup {\mathcal V}_2^{\ell, h_\ell} \doteq \{ i_1, i_2, \dots, i_c\},$$ 
where $c  \doteq {\rm card}({\mathcal V}_1^{\ell, h_\ell} \cup {\mathcal V}_2^{\ell, h_\ell})$ and $1 \le i_1 < i_2 < \dots < i_c \le N$, then 
$\eta_u^{\ell,h_\ell}(t)$ consists of the following blocks (in order)  $S_{i_1}^{\ell,h_\ell} e_{i_1}^{\ell,h_\ell}(t), S_{i_2}^{\ell,h_\ell} e_{i_2}^{\ell,h_\ell}(t), \dots S_{i_c}^{\ell,h_\ell} e_{i_c}^{\ell,h_\ell}(t),$ where 
$$
S_{i_k}^{\ell,h_\ell} e_{i_k}^{\ell,h_\ell}(t) \doteq  \begin{cases} e_{i_k}^{\ell,h_\ell}(t), & {\rm if} \ {\color{black}i_k}\in {\mathcal V}_1^{\ell, h_\ell};\\[2ex]
e_{i_k, 2u}^{\ell,h_\ell}(t),& {\rm if} \ {\color{black}i_k}\in {\mathcal V}_2^{\ell, h_\ell}.
\end{cases}
$$ 
Consequently
\begin{align} \label{eq.error2}
\eta_{u}^{\ell,h_\ell}(t+1)=\bar \Gamma_u^{\ell,h_\ell} \eta_{u}^{\ell,h_\ell}(t),
\end{align}
where
$\bar \Gamma_u^{\ell,h_\ell}$
takes the following form
\begin{equation}\bar \Gamma_u^{\ell,h_\ell} \doteq 
S^{\ell, h_\ell} 
\left(\left(I - k^{\ell,h_\ell}\mathcal{ L}^{\ell,h_\ell}\right)\otimes A^{\ell,h_\ell}\right)
(S^{\ell, h_\ell})^\top,
\label{defGammau}
\end{equation}
$\mathcal{L}^{\ell,h_\ell}$ is the  matrix obtained from the Laplacian ${\mathcal L}$ by deleting the rows and columns   indexed in  $\mathcal{V}_3^{\ell,h_\ell}$ and 
$$S^{\ell,h_\ell} \doteq {\rm diag} \{S_{i_1}^{\ell,h_\ell}, S_{i_2}^{\ell,h_\ell},\dots S_{i_c}^{\ell,h_\ell}\}.$$
We observe that the matrix $\Gamma_u$ in \eqref{eq.triang.und}, whose expression we had left temporarily undetermined, is equivalent up to a block permutation to the matrix
${\rm diag}\{ \bar \Gamma_u^{\ell,{h_\ell}}\}_{\color{black} (\ell,h_\ell): {\mathcal V}_3^{\ell,h_\ell} \subsetneq [1,N]},$
where
$\bar \Gamma_u^{\ell,h_\ell}$ is defined in \eqref{defGammau}. Therefore in order to determine under what conditions  $\eta_u(t)$ asymptotically converges to zero, or equivalently $\Gamma_u$ is Schur stable, we will identify under what conditions  each $\bar \Gamma_u^{\ell, h_\ell}$ is Schur stable.
When this is the case, there exists a distributed observer with the structure we have proposed.

The following theorem therefore provides   necessary and sufficient conditions for the existence of a distributed observer.
\medskip

\begin{theorem}\label{thmain}
    Consider the system described by \eqref{eq.sys} and \eqref{eq.out}, with directed communication graph $\mathcal{G}=(\mathcal{V},\mathcal{E}, \mathcal{A})$.  Let Assumptions \ref{ass.jointdet} and \ref{ass.obs} hold. There exists a distributed observer of the form \eqref{eq.luenberger} and \eqref{eq.estimate.undetectble.reduced}, for every $i\in [1,N]$,  if and only if  for each $(\ell,h_\ell)\in [1,r_u]\times [1, g_\ell]$  there exists $k^{\ell,h_\ell}\in\mathbb{R}$ such that
\begin{subequations}\label{eq.mainineq}
\begin{align}
    &|1 - k^{\ell,h_\ell}\mu|<\frac{1}{|\lambda_\ell|},\label{eq.mainineq.ineq}
    \\& \forall\ \mu \in\sigma\left(S^{\ell, h_\ell}\left(\mathcal{L}^{\ell,h_\ell}\otimes I_{d^{\ell,h_\ell}}\right)(S^{\ell, h_\ell})^\top\right).\label{eq.mainineq.eig}
\end{align}
\end{subequations}

\end{theorem}
\begin{proof}
We preliminarily observe that Assumption \ref{ass.obs} ensures that each pair $(F_{i,d}, H_{i,d}), i\in [1,N],$ is detectable.
Therefore, there exists a matrix $L_{i,d}$ that makes  $F_{i,d}-L_{i,d} H_{i,d}$  Schur, and hence $\Gamma_d$ in \eqref{eq.triang.und} is a Schur matrix. This guarantees that each observer \eqref{eq.luenberger} asymptotically estimates $z_{i,d}(t)$, $i\in [1, N]$. Therefore a distributed observer of the form \eqref{eq.luenberger} and \eqref{eq.estimate.undetectble.reduced} exists if and only if we are able to ensure that for every 
$(\ell,h_\ell)\in [1,r_u]\times [1, g_\ell]$ the estimation error in \eqref{eq.error2}  asymptotically converges to zero. But this is equivalent to guaranteeing the existence, for every $(\ell,h_\ell)\in [1,r_u]\times [1, g_\ell]$  of a real coefficient $k^{\ell, h_\ell}$ that makes the matrix
    $\bar \Gamma_u^{\ell,h_\ell} = 
S^{\ell, h_\ell} 
\left(\left(I - k^{\ell,h_\ell}\mathcal{ L}^{\ell,h_\ell}\right)\otimes A^{\ell,h_\ell}\right)
(S^{\ell, h_\ell})^\top$ Schur stable.
\\
By Lemma \ref{lem.1}, in the Appendix, this is possible if and only if  for every $(\ell,h_\ell)\in [1,r_u]\times [1, g_\ell]$  there exists $k^{\ell,h_\ell}\in\mathbb{R}$ 
such that \eqref{eq.mainineq} holds.
\end{proof}

\begin{remark} \label{explaineigenvalues}
   The spectrum of $S^{\ell, h_\ell}\left(\mathcal{L}^{\ell,h_\ell}\otimes I_{d^{\ell,h_\ell}}\right)$ $(S^{\ell, h_\ell})^\top$ coincides with the union of the spectra of all matrices 
    $\tilde S_{r_i} {\mathcal L}^{\ell, h_\ell} \tilde S_{r_i}^\top,$ for all $i\in [1, d^{\ell,h_\ell}]$, where   $\tilde S_{r_i}$ (see the proof of Lemma \ref{lem.1}) is the selection matrix including all row vectors ${\mathbb e}_j^\top$ corresponding  to the agents $j \in {\mathcal V}_1^{\ell, h_\ell} \cup {\mathcal V}_2^{\ell, h_\ell}$ for which the $i$th entry of $x^{\ell, h_\ell}(t)$ is undetectable (and, therefore $\tilde S_{r_1}=I$).
    This implies that condition \eqref{eq.mainineq.ineq} must be checked for all the eigenvalues of ${\mathcal L}^{\ell, h_\ell}$ and for all the eigenvalues of a family of principal submatrices of ${\mathcal L}^{\ell, h_\ell}$.
\end{remark}
\medskip

\section{Second Strategy}\label{sec.strategy2}
The strategy we have just explored for the distributed state estimation requires the existence of real parameters $k^{\ell, h_\ell}$ that satisfy 
condition \eqref{eq.mainineq.ineq} for a large set of eigenvalues, namely (see \eqref{eq.mainineq.eig}) all the eigenvalues of the matrix $S^{\ell,h_\ell}\left(\mathcal{L}^{\ell,h_\ell}\otimes I_{d^{\ell,h_\ell}}\right)(S^{\ell,h_\ell})^\top.$
We now explore an alternative strategy to design a distributed observer that will lead to the same condition as in \eqref{eq.mainineq.ineq} but for a typically much smaller family of eigenvalues namely those in 
$\sigma(\mathcal{L}^{\ell,h_\ell}) \subseteq
\sigma(S^{\ell,h_\ell}\left(\mathcal{L}^{\ell,h_\ell}\otimes I_{d^{\ell,h_\ell}}\right)(S^{\ell,h_\ell})^\top)$
(see Remark \ref{explaineigenvalues}).
The price to pay will be the need to resort to  augmented observers since the detectable and the undetectable parts of the state estimated by each agent will possibly overlap. Specifically, we will first estimate the detectable portion of the state
$z_{i,d}(t)$ using the same Luenberger observer \eqref{eq.luenberger} as in the first strategy, but we will retain only a portion of it (the one corresponding to the blocks $x^{\ell, h_\ell}(t)$ that are completely observable or correspond to stable eigenvalues). On the other hand, each $i$th agent will provide an estimate of the blocks $x^{\ell, h_\ell}(t)$ that are either completely unobservable or partly observable for it by using a consensus strategy. Differently from the observers \eqref{eq.estimate.undetectble.reduced}, only entire miniblocks $A^{\ell, h_\ell}$ will appear in the matrix governing the estimation error dynamics and not portions of it, say $S_i^{\ell, h_\ell} A^{\ell, h_\ell} (S_j^{\ell, h_\ell})^\top$. As an additional advantage, the dynamics of the estimation error relative to $x^{\ell, h_\ell}(t)$ will depend only on the estimation errors relative to the same block thus leading to a ``block diagonal structure" that we did not have in \eqref{eq.estimate.undetectble.reduced} (see $F_{i,*}$).  

%{\bf Alternative decomposition:}\ 
To derive the new distributed observer equations, we  introduce another permutation matrix, alternative to $Q_i$. After having 
performed Step 1 as in Section \ref{Sec.Reordering}, and hence
introduced the sets ${\mathbb G}_{i,k}^\ell, k\in[1,3],$ and the block diagonal matrices $A_{i,k}^\ell, k\in [1,3],$ for every $i\in [1,N]$ and every $
\ell\in [1,r_u],$  
we adopt this\footnote{\color{black}The new permuted form of  $A$ we will adopt for each agent $i$ differs from the one in \eqref{eq.ste3.perm} only in the way we treat the miniblocks $A^{\ell, h_\ell}$ corresponding to partly observable state variables, i.e., for $h_\ell\in {\mathbb G}_{i,2}^{\ell, h_\ell}$. Such form   could also be obtained from the one in \eqref{eq.ste3.perm} by  permutating the central blocks.} 

{\bf Alternative Step 2: Grouping of the blocks $A_{i,k}^{\ell}$ for the same value of $k
\in [1,3]$ and different values of $\ell\in [1,r_u]$.}\\
We let $R_i$ denote a permutation matrix that orders the unstable miniblocks of the matrix $A$ according to the sets $\mathbb{G}_{i,1}^{\ell}$, $\mathbb{G}_{i,2}^{\ell}$, and $\mathbb{G}_{i,3}^{\ell}$. 
This amounts to saying that
\begin{align}
{R_i}^\top AR_i&=\begin{bmatrix}
    A_{i,u}&\vline& {\mathbb 0}\\
    \hline
    \mathbb{0}&\vline& A_{i,d}
\end{bmatrix} = \begin{bmatrix}
    A_{i,1}&\mathbb{0} &\vline&\mathbb{0}&\mathbb{0}\\
    \mathbb{0}&A_{i,2}&\vline& \mathbb{0} &\mathbb{0}\\
    \hline
    \mathbb{0}&  \mathbb{0}&\vline& A_{i,3}&\mathbb{0}\\
        \mathbb{0}&  \mathbb{0}&\vline&\mathbb{0}&A_s
\end{bmatrix}, \nonumber\\
C_iR_i&=  \begin{bmatrix}C_{i,u}&\vline&C_{i,d}\end{bmatrix}=
 \begin{bmatrix}
    \mathbb{0}&C_{i,2}&\vline& C_{i,3}&C_{s}
\end{bmatrix}, \nonumber\\
{R_i}^\top B &=
\begin{bmatrix}
     B_{i,u}\\
     \hline
     B_{i,d}
 \end{bmatrix} = 
\begin{bmatrix}
    B_{i,1}\\
    B_{i,2}\\
    \hline
    B_{i,3}\\
    B_{s}
\end{bmatrix},  \nonumber
\end{align}
where $A_{i,2}\doteq{\rm diag}\{A_{i,2}^{\ell}\}_{\ell\in [1,r_u]}$, $ B_{i,2}\doteq{\rm col}\{{B_{i,2}^{\ell}}\}_{\ell\in [1,r_u]}$, and $ C_{i,2}\doteq{\rm row}\{{C_{i,2}^{\ell}}\}_{\ell\in [1,r_u]}$, while all the other blocks have been defined in Step 3.
This corresponds to splitting the (permuted) state vector as follows:
\begin{align}
{R_i}^\top x(t)=\begin{bmatrix}
    x_{i,u}(t)\\
    \hline
    x_{i,d}(t)
\end{bmatrix}\doteq\begin{bmatrix}
    x_{i,1}(t)\\
    x_{i,2}(t)\\
    \hline
    x_{i,3}(t)\\
    x_{i,s}(t)
\end{bmatrix},
\label{defxidu}
\end{align}
where  $x_{i,2}(t)\doteq{\rm col}\{{x_{i,2}^{\ell}(t)}\}_{\ell\in [1,r_u]}$ (see Step 3 for the other blocks). % \medskip

Agent $i$ will separately estimate $x_{i,d}(t)$ and $x_{i,u}(t)$ using two different observers.
The estimate $\hat x_{i,d}(t)$ will be obtained from the estimate  $\hat z_{i,d}(t)$ derived in  \eqref{eq.luenberger},
by exploiting the relationship between $x_{i,d}(t)$ and $z_{i,d}(t)$ (see \eqref{eq.zdzu} and \eqref{defxidu}):
\begin{equation}x_{i,d}(t)= 
    \begin{bmatrix}
        \mathbb{0}&\vline & I
    \end{bmatrix}
    \begin{bmatrix}
     x_{i,2o}(t)\\
     \hline
     x_{i,3}(t)\\
     x_{i,s}(t)
\end{bmatrix} = \begin{bmatrix}
        \mathbb{0}&\vline&I
    \end{bmatrix}  z_{i,d}(t).
    \label{eq:zdxd}
    \end{equation}
Meanwhile the estimate  $\hat x_{i,u}(t)$ will be obtained using the estimates   of the blocks $x^{\ell, h_\ell}(t)$ appearing in $x_{i,u}(t)$ provided by the neighbors of the $i$th agent. If we let $\hat x_{j}^{\ell, h_\ell}(t)$ denote the estimate that agent $j$ provides for the block $x^{\ell, h_\ell}(t)$, then for each $j \in {\mathcal N}_i$, the vector
$\hat x_{j}^{\ell, h_\ell}(t)$
may be part   either of $\hat x_{j,d}(t)$ or of $\hat x_{j,u}(t)$, depending on whether $h_\ell \in {\mathbb G}_{j,3}^\ell$ or $h_\ell \in {\mathbb G}_{j,1}^\ell \cup {\mathbb G}_{j,2}^\ell$ (equivalently, $j\in {\mathcal V}_3^{\ell, h_\ell}$ or $j\in {\mathcal V}_1^{\ell, h_\ell}\cup {\mathcal V}_2^{\ell, h_\ell}$).

Details are provided in the following. 
 The final estimate that agent $i$ provides of the state $x(t)$ will be
\begin{align*}
{\color{black}\hat x_i(t) = R_i \begin{bmatrix}
    \hat x_{i,u}(t)\cr\hline    
    \hat x_{i,d}(t)
\end{bmatrix}.}
\end{align*}
%\smallskip
{\bf Estimate of $x_{i,d}(t)$:}\
The
estimate of $x_{i,d}(t)$ provided  by agent $i$ relies on the Luenberger observer \eqref{eq.luenberger} and on the identity \eqref{eq:zdxd}, and is therefore described  by the following state-space  model:
\begin{subequations}\label{eq.lunberger.directed}
\begin{align}
    \hat z_{i,d}(t+1)=&F_{i,d}\hat z_{i,d}(t)+G_{i,d}u(t)\\
    &+L_{i,d}(y_i(t)-H_{i,d}\hat z_{i,d}(t)),\nonumber\\
    \hat x_{i,d}(t) =& \begin{bmatrix}
        \mathbb{0}&I
    \end{bmatrix} \hat z_{i,d}(t).\label{eq.lunberger.directed.b}
\end{align}
\end{subequations}
As noted previously, the detectability of the pair $(F_{i,d}, H_{i,d})$ ensures that there exists $L_{i,d}$ such that $F_{i,d}- L_{i,d} H_{i,d}$ is Schur stable, and hence $\eta_{i,d}(t)= z_{i,d}(t) - \hat z_{i,d}(t)$ asymptotically converges to $0$.  Therefore
$$e_{i,d}(t) \doteq x_{i,d}(t) - \hat x_{i,d}(t)= \begin{bmatrix}
        \mathbb{0}&I
    \end{bmatrix}  \eta_{i,d}(t)$$
asymptotically converges to $0$, in turn.

{\bf Estimate of $x_{i,u}(t)$:}\
We design an observer of the type
\begin{align}
   \hat  x_{i,u}(t+1)=&A_{i,u}\hat x_{i,u}(t)+B_{i,u}u(t) \label{eq.xiuhat}\\
&+K_iA_{i,u}\sum_{j\in\mathcal{N}_i}[\mathcal{A}]_{i,j}\begin{bmatrix}
        I&\mathbb{0}
    \end{bmatrix}R_i^\top(\hat x_j(t)-\hat x_i(t))
    \nonumber
\end{align}
 where  $K_i$ is a block-diagonal matrix to be designed in the following form:
 %{\small
 \begin{align*}
  K_i \doteq \begin{bmatrix}{\rm diag} \left\{ K_{i,1}^\ell\right\}_{\ell\in [1,r_u]} & \mathbb{0} \cr \mathbb{0}&{\rm diag} \left\{\bar K_{i,2}^\ell\right\}_{\ell\in [1,r_u]  }\end{bmatrix}
\end{align*} %}
where $K_{i,1}^\ell$ is defined as in \eqref{Kappa1}, while
\begin{align}
\bar K_{i,2}^\ell\doteq&\ {\rm diag} \left\{ k^{\ell,h_\ell} I_{d^{\ell, h_\ell}}\right\}_{ h_\ell \in {\mathbb G}_{i,2}^\ell}
\label{Kappa2bar}
\end{align}
 differs from \eqref{Kappa2}
 only in the dimensions ($d_{\ell,h_\ell} > t_i^{\ell, h_\ell} -1).$ 
This amounts to choosing a  gain $k^{\ell,h_\ell}$ for each Jordan miniblock $A^{\ell, h_\ell}$ (equivalently, for each block $x^{\ell, h_\ell}(t)$ and not for portions of it that vary with the agent), 
corresponding to some $(\ell, h_\ell)$ for which ${\mathcal V}_3^{\ell, h_\ell} \subsetneq [1,N].$
\\ In order to study the dynamics of the estimation error 
$e_{i,u} (t) \doteq x_{i,u}(t)-\hat x_{i,u}(t)$, we exploit the block diagonal structure of the matrices. Indeed, it is easy to deduce from \eqref{eq.xiuhat} that for every $(\ell, h_\ell)$ for which ${\mathcal V}_3^{\ell, h_\ell} \subsetneq [1,N]$, the  portion of the vector $\hat x_{i,u}(t)$ associated with the miniblock $A^{\ell,h_\ell}$, i.e., $\hat x_{i}^{\ell,h_\ell}(t)$, updates  according to\footnote{This equation is the same as \eqref{eq.estimate.undetectble.directed}. However, equation \eqref{eq.estimate.undetectble.directed} was valid only for $h_\ell\in {\mathbb G}_{i,1}^{\ell}$ and in addition it required to split each $
\hat x_j^{\ell, h_\ell}$ for $j\in {\mathcal N}_i \cap {\mathcal V}_2^{\ell, h_\ell},$ while as we will see this will no longer be necessary in this new set-up.}
\begin{align}
        \hat x_{i}^{\ell,h_\ell}&(t+1)=
        A^{\ell,h_\ell}
        \hat x_{i}^{\ell,h_\ell}(t)
+B^{\ell,h_\ell} u(t)\\&+k^{\ell,h_\ell}
        A^{\ell,h_\ell}\sum_{j\in\mathcal{N}_i}[\mathcal{A}]_{i,j}(\hat x_j^{\ell,h_\ell}(t)-\hat x_i^{\ell,h_\ell}(t)).\nonumber
\end{align}
Since the dynamics of the portion of the state $x^{\ell,h_\ell}(t)$ associated with the miniblock $A^{\ell,h_\ell}$ is
\begin{align}\label{eq.miniblock}
         x^{\ell,h_\ell}(t+1)=A^{\ell,h_\ell}x^{\ell,h_\ell}(t)+B^{\ell,h_\ell} u(t),
\end{align}
then the dynamics of the  estimation error $e_i^{\ell,h_\ell}(t) = x^{\ell,h_\ell}(t)-\hat x_i^{\ell,h_\ell}(t)$ is
\begin{align}
     e_{i}^{\ell,h_\ell}&(t+1)=A^{\ell,h_\ell}
         e_{i}^{\ell,h_\ell}(t)\\
         &- k^{\ell,h_\ell}
        A^{\ell,h_\ell}\sum_{j\in \mathcal{ N}_{i}}[\mathcal{A}]_{i,j}(e_i^{\ell,h_\ell}(t)-e_j^{\ell,h_\ell}(t)),
        \nonumber
\end{align}
and 
we can rewrite the previous equation as 
\begin{align}\label{eq.est.err.lap.dir}
     e_{i}^{\ell,h_\ell}(t+1)=&(1 - k^{\ell,h_\ell}[\mathcal{L}]_{i,i})A^{\ell,h_\ell}
         e_{i}^{\ell,h_\ell}(t)\\
         &- k^{\ell,h_\ell}
        A^{\ell,h_\ell}\sum_{j\in \mathcal{ N}_{i}}[\mathcal{L}]_{i,j}e_j^{\ell,h_\ell}(t).\nonumber
\end{align}
The dynamics of the error $e_{i}^{\ell,h_\ell}(t)$ of the $i$th agent is influenced by the estimation errors $e_j^{\ell,h_\ell}(t)$ of its neighbors ($j\in {\mathcal N}_i$), which can be split into two categories: those for which the $h_\ell$th miniblock is completely observable, namely  $j \in \mathcal{V}_{3}^{\ell,h_\ell}$, and the others. 
Therefore,  the sum in \eqref{eq.est.err.lap.dir} can be rewritten as 
    \begin{align}\label{eilh.aug}
        e_{i}^{\ell,h_\ell}(t+1)&=(1 -  k^{\ell,h_\ell}[\mathcal{L}]_{i,i})
         A^{\ell,h_\ell}e_{i}^{\ell,h_\ell}(t)\\
         &- k^{\ell,h_\ell}
        A^{\ell,h_\ell}\sum_{j\in {\mathcal N}_i \cap(\mathcal{V}_{1}^{\ell,h_\ell} \cup \mathcal{V}_{2}^{\ell,h_\ell})}[\mathcal{L}]_{i,j}e_j^{\ell,h_\ell}(t) \nonumber
        \\&- k^{\ell,h_\ell}
        A^{\ell,h_\ell}\sum_{j\in {\mathcal N}_i \cap\mathcal{V}_{3}^{\ell,h_\ell}}[\mathcal{L}]_{i,j}e_j^{\ell,h_\ell}(t).\nonumber
\end{align}
We observe that if we define {\color{black}$e_u(t) \doteq {\rm col}\{e_{i,u}(t)\}_{i\in [1,N]}$, $e_d(t) \doteq {\rm col}\{e_{i,d}(t)\}_{i\in [1,N]}$, $\eta_d(t) \doteq {\rm col}\{\eta_{i,d}(t)\}_{i\in [1,N]}$}, 
then the dynamics of the estimation error 
{\color{black}$[e^\top_u(t)\ e^\top_d(t)]^\top$} 
updates according to an equation of the following form:
\begin{align}
\label{eq.triang}
\begin{bmatrix} e_u(t+1)\cr \eta_d(t+1)\end{bmatrix}
&= \begin{bmatrix} \Phi_u & \Phi_*S \cr {\mathbb 0} & \Gamma_d\end{bmatrix}
\begin{bmatrix} e_u(t)\cr \eta_d(t)\end{bmatrix}\\
\begin{bmatrix}
    e_u(t)\\
    e_d(t)
\end{bmatrix}&=\begin{bmatrix}
    I&\mathbb{0}\\ \mathbb{0}&S
\end{bmatrix}\begin{bmatrix} e_u(t)\cr \eta_d(t)\end{bmatrix}.
\end{align}
where $\Gamma_d \doteq {\rm diag} \{ F_{i,d} - L_{i,d}H_{i,d}\}_{i\in [1,N]}$ is Schur,   $S$ is a selection matrix\footnote{It is actually a block diagonal matrix whose diagonal blocks are the matrices $\begin{bmatrix} {\mathbb 0} & I\end{bmatrix}$ appearing in \eqref{eq:zdxd}.}, $\Phi_*$ and $\Phi_u$ can be obtained by using of \eqref{eilh.aug}. Note that the expressions of $S$ and $\Phi_*$ are irrelevant for the subsequent discussion. We will come later to the expression of $\Phi_u$.\\
The convergence to zero of $\eta_{d}(t)$ and consequently of $e_{d}(t)$  (and, hence, of  every $e_{i,d}(t)$) has already been discussed. On the other hand, due to the block triangular structure of the matrix in \eqref{eq.triang}, 
condition 
 \begin{align}
 \label{eq.tuttinulli}
    \lim_{t\to+\infty}e_{i}^{\ell,h_\ell}(t)=\mathbb{0},\quad \forall\ i\in[1,N]
 \end{align}
holds 
for every $(\ell, h_\ell) \in  [1,r_u]\times  [1,g_\ell]$ and 
$i\in\mathcal{V}_1^{\ell,h_\ell}\cup \mathcal{V}_2^{\ell,h_\ell}$ 
if and only if 
the matrix $\Phi_u$ is Schur, which is equivalent to the asymptotic stability,  for every $(\ell, h_\ell) \in  [1,r_u]\times  [1,g_\ell]$ and $i\in\mathcal{V}_1^{\ell,h_\ell}\cup \mathcal{V}_2^{\ell,h_\ell}$, of the  error dynamics
\begin{align}
       e_{i}^{\ell,h_\ell}(t+1) =&(1 - k^{\ell,h_\ell}[\mathcal{L}]_{i,i})
         A^{\ell,h_\ell}e_{i}^{\ell,h_\ell}(t)\\
         &- k^{\ell,h_\ell}
        A^{\ell,h_\ell}
        \!\!\!\!\!\!\sum_{j\in {\mathcal N}_i \cap(\mathcal{V}_{1}^{\ell,h_\ell} \cup \mathcal{V}_{2}^{\ell,h_\ell})} \!\!\![\mathcal{L}]_{i,j}e_j^{\ell,h_\ell}(t), \nonumber
\end{align} 
or, equivalently, of the linear system
\begin{align} \label{eq.error}
e_{u}^{\ell,h_\ell}(t+1)=\left(\left(I - k^{\ell,h_\ell}\mathcal{ L}^{\ell,h_\ell}\right)\otimes A^{\ell,h_\ell}\right) e_{u}^{\ell,h_\ell}(t),
\end{align}
where  
$e_{u}^{\ell, h_\ell} \doteq {\rm col} \{ e_i^{\ell, h_\ell}(t)\}_{i \in\mathcal{V}_1^{\ell,h_\ell}\cup \mathcal{V}_2^{\ell,h_\ell}}$ and 
 $\mathcal{L}^{\ell,h_\ell}$ is the  matrix obtained from the Laplacian ${\mathcal L}$ by deleting the rows and columns   indexed in  $\mathcal{V}_3^{\ell,h_\ell}$.

The following theorem provides    necessary and sufficient conditions for the asymptotic stability of   system \eqref{eq.error} for all  $(\ell, h_\ell) \in  [1,r_u]\times  [1,g_\ell]$, and hence for the existence of a distributed observer of the form \eqref{eq.lunberger.directed} and \eqref{eq.xiuhat}.
\medskip

\begin{theorem}\label{thmain.aug}
    Consider a system described by \eqref{eq.sys} and \eqref{eq.out}, with directed communication graph $\mathcal{G}=(\mathcal{V},\mathcal{E}, \mathcal{A})$.  Let Assumptions \ref{ass.jointdet} and \ref{ass.obs} hold. There exists a distributed (augmented) observer of the form \eqref{eq.lunberger.directed} and \eqref{eq.xiuhat}, for every $i\in [1,N]$,  if and only if  for all $(\ell,h_\ell)\in [1,r_u]\times [1, g_\ell]$  there exists $k^{\ell,h_\ell}\in\mathbb{R}$ such that
    \begin{subequations}    \label{eq.mainineq.aug}
    \begin{align}
    &|1 - k^{\ell,h_\ell}\mu|<\frac{1}{|\lambda_\ell|},    \label{eq.mainineq.aug.ineq}
  \\& \forall\ \mu \in\sigma\left(\mathcal{L}^{\ell,h_\ell}\right).    \label{eq.mainineq.aug.eig}
    \end{align}
    \end{subequations}
\end{theorem}
\begin{proof}
We preliminarily observe that Assumption \ref{ass.obs} ensures that each pair $(F_{i,d}, H_{i,d}), i\in [1,N],$ is detectable.
Therefore, there exists a matrix $L_{i,d}$ that makes  $F_{i,d}-L_{i,d} H_{i,d}$  Schur, and hence $\Gamma_d$ is a Schur matrix. This guarantees that each observer \eqref{eq.lunberger.directed} asymptotically estimates $z_{i,d}(t)$ and therefore $x_{i,d}(t)$ for every $i\in [1, N]$. Therefore a distributed observer of the form \eqref{eq.lunberger.directed} and \eqref{eq.xiuhat} exists if and only if we are able to ensure that for every 
$(\ell,h_\ell)\in [1,r_u]\times [1, g_\ell]$ the estimation error in \eqref{eq.error} asymptotically converges to zero. But this is equivalent to guaranteeing the existence, for every $(\ell,h_\ell)\in [1,r_u]\times [1, g_\ell]$,  of a real coefficient $k^{\ell, h_\ell}$ that makes the matrix
    $\left(I -  k^{\ell,h_\ell}\mathcal{ L}^{\ell,h_\ell}\right)\otimes A^{\ell,h_\ell}$ Schur stable.
By Corollary \ref{cor.dir}, in the Appendix, this is possible if and only if  for every $(\ell,h_\ell)\in [1,r_u]\times [1, g_\ell]$  there exists $k^{\ell,h_\ell}\in\mathbb{R}$ 
such that \eqref{eq.mainineq.aug} holds.
\end{proof}

\section{Comparisons } \label{sec.comparison}
At this point we want to compare in detail the distributed observers proposed in this paper with the original in \cite{GaoYang}.

The main advantage of both  distributed observers   is that they are simpler than the  one in \cite{GaoYang} since they both focus on the estimate of 
 subvectors of the state vector corresponding either to entire Jordan miniblocks or to portions of it, rather than considering the estimates of single entries (or pairs of entries in case of  complex eigenvalues).
As a consequence of this choice, the solvability conditions of the problem are considerably simplified for two main reasons. 
First, the number of inequalities to be verified is greatly reduced: instead of one inequality per entry, we now have a single inequality for each miniblock (condition \eqref{eq.mainineq} for the first strategy and  \eqref{eq.mainineq.aug} for the second strategy). This also means that the number of matrices  ($S^{\ell, h_\ell}\left(\mathcal{L}^{\ell,h_\ell}\otimes I_{d^{\ell,h_\ell}}\right)$ $(S^{\ell, h_\ell})^\top$  or ${\mathcal L}^{\ell,h_\ell}$, depending on the strategy) whose spectra must be computed is significantly lower. Second, rather than imposing a single coupling strength $k$
 for all observers, we can relax the existence conditions by allowing a distinct coupling strength $k^{\ell,h_\ell}$ for each pair 
$(\ell, h_\ell)$ (see footnote \ref{footn:numerok}). For all these reasons, conditions \eqref{eq.mainineq}  and  \eqref{eq.mainineq.aug} are less restrictive than condition (30) in \cite{GaoYang}.

Finally, condition C2 in Theorem 2 of \cite{GaoYang} proves to be implicit in \eqref{eq.mainineq}  and  \eqref{eq.mainineq.aug}. Indeed, the existence of $k^{\ell, h_\ell}$ satisfying either of those conditions rules out the possibility that $\mu=0$ is an eigenvalue of ${\mathcal L}^{\ell, h_\ell}$. So, conditions \eqref{eq.mainineq}  and  \eqref{eq.mainineq.aug}  hold true only if (see Lemma 4 in \cite{GaoYang}) for every $(\ell, h_\ell)$ for which ${\mathcal V}_3^{\ell, h_\ell} \subsetneq[1,N]$ it holds what follows: every node indexed in ${\mathcal V}_1^{\ell, h_\ell} \cup {\mathcal V}_2^{\ell, h_\ell}$ is reachable from at least one of the nodes indexed in ${\mathcal V}_3^{\ell, h_\ell}$ by means of a   directed path, namely there exists a ``directed spanning forest" rooted in the agents indexed in ${\mathcal V}_3^{\ell, h_\ell}$.
However, we remark again that compared to condition C2 in Theorem 2 of \cite{GaoYang} this communication constraint is imposed on a (typically much) smaller number of sets.

While the previous comments apply to both strategies,   they differ in the way they treat the portions of the state vectors that correspond to 
miniblocks $A^{\ell, h_\ell}$ that are partly observable for an agent $i$.
Indeed,   the  distributed observer \eqref{eq.lunberger.directed}-\eqref{eq.xiuhat} proposed in Section \ref{sec.strategy2} focuses on the estimate of the subvectors of the state vector corresponding to entire Jordan miniblocks. 
    So, the main distinction is between agents that can estimate the whole subvector $x^{\ell, h_\ell}(t)$ and agents that cannot. This is the   criterion that leads to split, for each $i$th agent, the estimate $\hat x_i(t)$ in the two parts $\hat x_{i,d}(t)$ and  $\hat x_{i,u}(t)$. However, $\hat x_{i,d}(t)$ cannot be estimated independently, but is derived from 
  $\hat z_{i,d}(t)$.   
    This leads to  an augmented state for the observers, as the distributed observer proposed in \cite{GaoYang}.
Indeed,
\begin{align*}
&{\rm dim} (\hat z_{i,d}(t)) + {\rm dim} (\hat x_{i,u}(t))\\
&=n + ({\dim}(\hat z_{i,d}(t))- {\dim}(\hat x_{i,d}(t)) \\
&= n + \sum_{\ell=1}^{r_u}  \sum_{h_\ell\in {\mathbb G}_{i,2}^{\ell, h_\ell}} (d^{\ell, h_\ell} - t_i^{\ell, h_\ell} +1),
\end{align*}
but this brings the advantage that condition \eqref{eq.mainineq.aug.ineq} needs to be satisfied only for  the eigenvalues $\mu$ of 
$\mathcal{L}^{\ell,h_\ell}$ and
    \begin{align*}
       {\rm dim}(\mathcal{L}^{\ell,h_\ell})
       ={\rm card} (\mathcal{V}_1^{\ell,h\ell}) + {\rm card} (\mathcal{V}_2^{\ell,h\ell})\le N-1.
    \end{align*}
On the other hand,
Strategy 1 
has the advantage that  the observers that the agents implement have
the same dimension as the original state (and in that respect it is preferable to the previous one and the one in \cite{GaoYang}), since
${\rm dim} (\hat z_{i,d}(t)) + {\rm dim} (\hat z_{i,u}(t)) =n$.

    The price to pay is that the existence of an asymptotic distributed observer of this type requires condition \eqref{eq.mainineq.ineq} (equivalently
\eqref{eq.mainineq.aug.ineq}) to be satisfied for all the eigenvalues $\mu$ of all the submatrices
$\tilde S_{r_i} {\mathcal L}^{\ell, h_\ell} \tilde S_{r_i}^\top$ (see Remark \ref{explaineigenvalues}). 
In detail, it is possible to prove (but we omit the lengthy details) that  for each $(\ell,h_\ell)\in[1,r_u]\times [1, g_\ell]$, the number of eigenvalues for which \eqref{eq.mainineq.ineq} needs to be verified is (in the worst case) equal to 
$$\sum_{m=1}^{d^{\ell,h_\ell}} (N-m) = N d^{\ell,h_\ell} - \frac{d^{\ell,h_\ell}(d^{\ell,h_\ell}+1)}{2}.$$
This number is still significantly smaller than the number of inequalities that need to be checked in condition (30) of \cite{GaoYang}.

From a practical point of view, 
    the natural approach to the choice of which distributed observer one may want to implement is to  first verify if conditions \eqref{eq.mainineq} hold for every $(\ell, h_\ell) \in [1,r_u]\times [1,g_\ell]$. If, so a distributed non-augmented observer is available. 
    If this is not possible, one may try to adopt Strategy 2 by checking if the weaker conditions \eqref{eq.mainineq.aug} hold for every $(\ell, h_\ell) \in [1,r_u]\times [1,g_\ell]$ and hence design  augmented distributed observers.

\section{Undirected Graph Case}\label{sec.und}

If we now relax our assumption on the communication graph and suppose that $\mathcal{G}=(\mathcal{V},\mathcal{E}, \mathcal{A})$ is an undirected graph, and hence ${\mathcal L}$ and all its principal submatrices ${\mathcal L}^{\ell, h_\ell}$ are symmetric, we can exploit the Cauchy Interlacing Theorem \cite{Hwang01022004} 
and claim that condition \eqref{eq.mainineq.ineq} is  satisfied for all the eigenvalues $\mu$ of all the submatrices
$\tilde S_{r_i} {\mathcal L}^{\ell, h_\ell} \tilde S_{r_i}^\top$ (see Remark \ref{explaineigenvalues}) if and only if such condition is satisfied for all the eigenvalues $\mu$ of 
${\mathcal L}^{\ell, h_\ell}$.
This means that in the case of an undirected communication graph it is always convenient to resort to the first strategy, since the dimension of the distributed observers is kept at its minimum value, and at the same time the solvability condition coincides with the simpler one derived for the second strategy  (namely, \eqref{eq.mainineq.aug}).

There is no loss of generality in assuming that $\mathcal{G}$ is connected. If not, the analysis we provide in the sequel needs to be applied to each single connected component (and for each such component Assumptions \ref{ass.jointdet}
and \ref{ass.obs} must hold).
As a consequence, the Laplacian  ${\mathcal L}$ is an irreducible symmetric matrix, and since Assumption \ref{ass.jointdet} ensures that 
for every $(\ell, h_\ell)\in [1,r_u]\times [1, g_\ell]$ the set
 $\mathcal{V}_3^{\ell,h_\ell}$ is not empty, it follows that each proper principal submatrix ${\mathcal L}^{\ell, h_\ell}$ of ${\mathcal L}$ has all positive real eigenvalues.
This allows to further 
simplify condition
 \eqref{eq.mainineq.aug}. 
 \medskip

 \begin{theorem} \label{th.main.und}
Consider a system described by \eqref{eq.sys} and \eqref{eq.out}, with undirected communication graph $\mathcal{G}=(\mathcal{V},\mathcal{E}, \mathcal{A})$.  Let Assumptions \ref{ass.jointdet} and \ref{ass.obs} hold. There exists a distributed observer of the form \eqref{eq.luenberger} and \eqref{eq.estimate.undetectble.reduced}, for every $i\in [1,N]$, if and only if  for all $(\ell,h_\ell)\in [1,r_u]\times [1, g_\ell]$  
    \begin{align}\label{eq.ineq.kset}
\frac{\mu_{M}^{\ell, h_\ell}}{\mu_{m}^{\ell, h_\ell}}
< 
\frac{|\lambda_\ell| +1}{|\lambda_\ell|-1},
\end{align}
where $0 < \mu_{m}^{\ell, h_\ell} \le \mu_{M}^{\ell, h_\ell}$ are the minimum and maximum (real) eigenvalues of $\mathcal{L}^{\ell,h_\ell}$.
If so, for each $(\ell,h_\ell)\in [1,r_u]\times [1, g_\ell]$  the coupling gain $k^{\ell,h_\ell}\in\mathbb{R}$ must be chosen so that
    \begin{align}\label{eq.kset}
k^{\ell, h_\ell}\in\left(\frac{1}{\mu_{m}^{\ell, h_\ell}}\left(1-\frac{1}{|\lambda_\ell|}\right),\frac{1}{\mu_{M}^{\ell, h_\ell}}\left(1+\frac{1}{|\lambda_\ell|}\right)\right).
\end{align}
\end{theorem}
\begin{proof} To prove this result it is sufficient to prove that, when ${\mathcal G}$ is undirected and connected, condition \eqref{eq.mainineq}
in Theorem \ref{thmain} becomes condition \eqref{eq.kset}, which makes sense (meaning that the interval is nonempty) if and only if 
\eqref{eq.ineq.kset} holds.
But this result follows from Corollary \ref{cor.und} in the Appendix.
\end{proof}

Note that
if the graph ${\mathcal G}$ is undirected, in order to verify if the problem is solvable
we need to check a single
condition, namely  \eqref{eq.ineq.kset}, for each $(\ell,h_\ell)\in[1,r_u]\times [1, g_\ell]$, while when ${\mathcal G}$ is a  directed graph the conditions to check - in the worst case - are as many as   ${\rm card}({\mathcal L}^{\ell, h_\ell})$ in case of Strategy 2, and ${\rm card}(S^{\ell,h_\ell} ({\mathcal L}^{\ell, h_\ell}\otimes I_{d^{\ell, h_\ell}}) (S^{\ell,h_\ell})^\top))$ in case of Strategy 1.

\section{Numerical Example}\label{sec.example}
\begin{figure}[tp!]
  \centering
\includegraphics[width=0.18\textwidth]{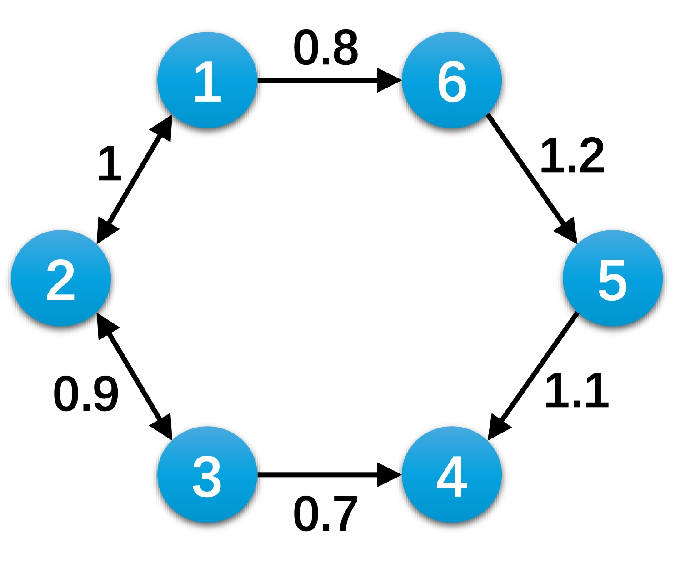}
  \caption{The communication digraph.}\label{graph}
\end{figure}

In this {\color{black} section}, an example is provided to show that the effectiveness of the proposed strategies. 
{\color{black} Consider  a $9$-dimensional  system described as in \eqref{eq.sys}, and assume  that the system matrix $A$ is in Jordan form \eqref{Jform}-\eqref{Jblock}-\eqref{Jminiblock},
with three Jordan miniblocks corresponding to a single  (unstable) eigenvalue $\lambda_1$ (and hence $A = A^1$),  i.e.,}
\begin{equation*}
\renewcommand{\arraystretch}{1}
A^{1,1}=\scriptsize\begin{bmatrix}1&1&0\\0&1&1\\0&0&1\end{bmatrix},
A^{1,2}=\begin{bmatrix}1&0\\0&1\end{bmatrix},
A^{1,3}=\begin{bmatrix}1&1&0&0\\0&1&1&0\\0&0&1&1\\0&0&0&1\end{bmatrix}.
\end{equation*}
%%%%%%%%%%%%%%%
%%%%%%%%%%%%%%%
The known matrix $B$ is omitted, as it does not affect the estimation results. Suppose that the system state $x(t)$ is estimated {\color{black} by a network of $6$ sensor nodes} connected through the directed graph in Fig. \ref{graph}, where the weights of the edges are given. According to the block-partition of $A$, the output matrices $C_i$, $i\in[1,6]$, are in the form of $[C_i^{1,1}\ C_i^{1,2}\ C_i^{1,3}]$ with
\begin{equation*}
\renewcommand{\arraystretch}{1}
\begin{split}
C_1^{1,1}=\scriptsize\begin{bmatrix}1&0&0\\0&0&0\\0&0&0\end{bmatrix},
C_1^{1,2}=\scriptsize\begin{bmatrix}0&0\\0&1\\0&0\end{bmatrix},
C_1^{1,3}=\scriptsize\begin{bmatrix}0&0&0&0\\0&0&0&0\\0&0&0&0\end{bmatrix},\\
C_2^{1,1}=\scriptsize\begin{bmatrix}0&1&0\\0&0&0\\0&0&0\end{bmatrix},
C_2^{1,2}=\scriptsize\begin{bmatrix}0&0\\1&0\\0&0\end{bmatrix},
C_2^{1,3}=\scriptsize\begin{bmatrix}0&0&0&0\\0&0&0&0\\0&1&0&0\end{bmatrix},\\
C_3^{1,1}=\scriptsize\begin{bmatrix}0&0&0\\0&0&0\\0&0&0\end{bmatrix},
C_3^{1,2}=\scriptsize\begin{bmatrix}0&0\\0&1\\0&0\end{bmatrix},
C_3^{1,3}=\scriptsize\begin{bmatrix}0&0&0&0\\0&0&0&0\\1&0&0&0\end{bmatrix},\\
C_4^{1,1}=\scriptsize\begin{bmatrix}0&0&1\\0&0&0\\0&0&0\end{bmatrix},
C_4^{1,2}=\scriptsize\begin{bmatrix}0&0\\0&0\\0&0\end{bmatrix},
C_4^{1,3}=\scriptsize\begin{bmatrix}0&0&0&0\\0&0&0&0\\0&0&1&0\end{bmatrix},\\
C_5^{1,1}=\scriptsize\begin{bmatrix}0&0&0\\0&0&0\\0&0&0\end{bmatrix},
C_5^{1,2}=\scriptsize\begin{bmatrix}0&0\\1&0\\0&0\end{bmatrix},
C_5^{1,3}=\scriptsize\begin{bmatrix}0&0&0&0\\0&0&0&0\\0&0&0&0\end{bmatrix},\\
C_6^{1,1}=\scriptsize\begin{bmatrix}1&0&0\\0&0&0\\0&0&0\end{bmatrix},
C_6^{1,2}=\scriptsize\begin{bmatrix}0&0\\0&1\\0&0\end{bmatrix},
C_6^{1,3}=\scriptsize\begin{bmatrix}0&0&0&0\\0&0&0&0\\1&0&0&0\end{bmatrix}.
\end{split}
\end{equation*}
%%%%%%%%%%%%%%%%%
%%%%%%%%%%%%%%%% 
{\color{black}The sets
$\mathcal{V}_{k}^{1,h_1}$, with $k\in [1,3]$ and $h_1\in[1,3]$, $i\in[1,6]$,} defined in \eqref{eq:defVlhl} are identified as $\mathcal{V}_1^{1,1}=\{3,5\}$, $\mathcal{V}_2^{1,1}=\{2,4\}$, $\mathcal{V}_3^{1,1}=\{1,6\}$,  $\mathcal{V}_1^{1,2}=\{4\}$, $\mathcal{V}_2^{1,2}=\{1,3,6\}$, $\mathcal{V}_3^{1,2}=\{2,5\}$, $\mathcal{V}_1^{1,3}=\{1,5\}$, $\mathcal{V}_2^{1,3}=\{2,4\}$, $\mathcal{V}_3^{3,6}=\{1,5\}$. 
 %%%%%%%%%%
Then we  obtain the matrices $S_i^{1,h_1}$, $ i\in\mathcal{V}_1^{1,h_1}\cup\mathcal{V}_2^{1,h_1}$, $h_1\in[1,3]$   as 
$S_1^{1,2}=[1\ 0]$, $S_1^{1,3}=I_4$,  $S_2^{1,1}=[1\ \mathbb{0}_{1\times 2}]$, 
$S_2^{1,3}=[1\ \mathbb{0}_{1\times 3}]$, $S_3^{1,1}=I_3$, $S_3^{1,2}= [1\ 0]$, 
$S_4^{1,1}=[I_2\ \mathbb{0}_{2\times 1}]$, $S_4^{1,2}=  I_2$, 
$S_4^{1,3}=[I_2\ \mathbb{0}_{2\times 2}]$, $S_5^{1,1}= I_3$, 
$S_5^{1,3}= I_4$,  $S_6^{1,2}=[1\ 0]$, and   $\mathcal{L}^{1,h_1}$,  $h_1\in[1,3]$  as follows
\begin{equation*}
\renewcommand{\arraystretch}{1}
\begin{split}
\mathcal{L}^{1,1}=&\ \scriptsize\begin{bmatrix}
1.9& -0.9&  0 & 0\\
-0.9& 0.9&  0 & 0\\
 0 & -0.7&1.8&  -1.1\\
0 & 0 & 0 & 1.2\end{bmatrix}
 \mathcal{L}^{1,2}= \scriptsize\begin{bmatrix}
1 &  0  & 0 & 0\\
0 &  0.9& 0  &  0\\
0 & -0.7&1.8&0\\
-0.8&0&0&0.8\end{bmatrix}\\
\mathcal{L}^{1,3}=&\ \scriptsize\begin{bmatrix}
1& -1& 0&0\\
 -1& 1.9& 0&0\\
 0& 0&1.8&-1.1\\
 0 & 0&0&1.2\end{bmatrix}.
 \end{split}
\end{equation*}
%%%%%%%%%%%%
%%%%%%%%%%%
\textbf{Strategy 1.} Following the procedure in Section \ref{sec.strategy1}, we need to first obtain $F_{i,d}$ and $H_{i,d}$, $i\in[1,6]$ in \eqref{eq.luenberger}. From $A^{1,h_1}$ and $C_i^{1,h_1}$, we have \\
$F_{1,d}=\mathrm{diag}\{A_{1,o}^{1,2},A^{1,1}\}$,
$H_{1,d}=[C_{1,o}^{1,2}\ C_1^{1,1}]$, 
$F_{2,d}=\mathrm{diag}\{A_{2,o}^{1,1},A_{2,o}^{1,3},A^{1,2}\}$,
$H_{2,d}=[C_{2,o}^{1,1}\ C_{2,o}^{1,3}\ C_2^{1,2}]$, 
$F_{3,d}=\mathrm{diag}\{A_{3,o}^{1,2},A^{1,3}\}$,
$H_{3,d}=[C_{3,o}^{1,2}\ C_3^{1,3}]$, 
$F_{4,d}=\mathrm{diag}\{A_{4,o}^{1,1},$ $A_{4,o}^{1,3}\}$,
$H_{4,d}=[C_{4,o}^{1,1}\ C_{4,o}^{1,3}]$, 
$F_{5,d}=A^{1,2}$, $H_{5,d}=C_5^{1,2}$,
$F_{6,d}=\mathrm{diag}\{A_{6,o}^{1,2},A^{1,1},A^{1,3}\}$,
$H_{6,d}=[C_{6,o}^{1,2}\ C_6^{1,1}\ C_6^{1,3}]$, \\
where \\
$A_{1,o}^{1,2}=1$, $C_{1,o}^{1,2}=[0;1;0]$, 
$A_{2,o}^{1,1}=[1\ 1;0\ 1]$, $C_{2,o}^{1,1}=[1\ 0;0\ 0;0\ 0]$, 
$A_{2,o}^{1,3}=[1\ 1\ 0;0\ 1\ 1;0\ 0\ 1]$, $C_{2,o}^{1,3}=[0\ 0\ 0;0\ 0\ 0;1\ 0\ 0]$, 
$A_{3,o}^{1,2}=1$, $C_{3,o}^{1,2}=[0;1;0]$, 
$A_{4,o}^{1,1}=1$, $C_{4,o}^{1,1}=[1;0;0]$, 
$A_{4,o}^{1,3}=[1\ 1;0\ 1]$, $C_{4,o}^{1,3}=[0\ 0;0\ 0;1\ 0]$, 
$A_{6,o}^{1,2}=1$, $C_{6,o}^{1,2}=[0;1;0]$. 
Then the following design matrix parameters  $L_{i,d}$, $i\in[1,6]$, in \eqref{eq.luenberger}  are selected to ensure that  $F_{i,d} - L_{i,d} H_{i,d}$ is Schur stable
\begin{equation*}
\renewcommand{\arraystretch}{1}
\begin{split}
L_{1d}=&\ \begin{bmatrix}
0.001&   0.3526    &     0\\
    1.0474  &  3.9058     &    0\\
    0.3444  &  2.7145     &    0\\
    0.0351 &   0.4125    &     0
\end{bmatrix}\!\!, \
 L_{3d}=  \begin{bmatrix}
 0  &   0.35&    0\\
         0  &   2.2326   &  1.4\\
         0  &   2.3361  &   0.71\\
         0  &   0.7692 &    0.154\\
         0  &   0.0789  & 
           0.012
         \end{bmatrix},\\
L_{4d}=&\  \begin{bmatrix}
         0.3005  & 0   & 0.0012\\
    0.1047   &      0    &0.5995\\
    0.0274   &      0   & 0.08
    \end{bmatrix}\!\!, \
   L_{5d}=  \begin{bmatrix}
0    &  0.7 &      0\\
         0   &   0.12 &       0
         \end{bmatrix},\\
 L_{2d}=&\ \begin{bmatrix}
0.5496 &   0&  0.0041\\
    0.071&    0&  0.0011\\
    0.2481&  0&   1.0404\\
    0.1616&0&  0.3406\\
    0.0244  &0&  0.0351\\
    0 &   0.7&0  \\
    0& 0.12&0
\end{bmatrix}\!\!, \
 L_{6d}=\scriptsize \begin{bmatrix}
 0.0001  &  0.3107   &0\\
    0.9169  &  2.7374   & 0.0032\\
    0.2725  &  1.7039   & 0.0019\\
    0.0262  &  0.2503  &  0.0003\\
    0.2062   & 1.923   & 1.3524\\
    0.1967   & 2.417   & 0.6629\\
    0.0602  &  0.9194   & 0.1396\\
    0.0059   & 0.1061  &  0.0106
      \end{bmatrix}\!\!.
\end{split}
\end{equation*}
%%%%%%%%%%%%%
%%%%%%%%%%%%%
Next, we need to determine the design parameters  $k^{1,h_1}$, $h_1\in[1,3]$. From $S_i^{1,h_1}$ and $\mathcal{L}^{1,h_1}$, $ i\in\mathcal{V}_1^{1,h_1}\cup\mathcal{V}_2^{1,h_1}$, $h_1\in[1,3]$, we  compute {\color{black} the spectra $\sigma(S^{1, 1}\left(\mathcal{L}^{1,1}\otimes I_{d^{1,1}}\right) (S^{1, 1})^\top)=\{1.8,0.3704, 2.4296, 0.9,$ $1.2\}$, 
$\sigma(S^{1, 2}\left(\mathcal{L}^{1,2}\otimes I_{d^{1,2}}\right) (S^{1, 2})^\top)$ $=\{ 0.8,1, 1.8, 0.9\}$, and
$\sigma((S^{1, 3}\left(\mathcal{L}^{1,3}\otimes I_{d^{1,3}}\right) (S^{1, 3})^\top)= \{ 0.3534,2.5466,1,$ $1.8,1.2\}$. }
%%%%%%%%%%%%%%%%
Then, according to \eqref{eq.mainineq}, we need to  solve $5$ inequalities, $4$ inequalities, and $5$  inequalities to obtain $k^{1,1}$, $k^{1,2}$, and $k^{1,3}$, respectively. Specifically,  $0<k^{1,1}<0.8232$, $0<k^{1,2}<1.1111$, $0<k^{1,1}<0.7854$. We have set $k^{1,1}=0.5$, $k^{1,2}=1$, $k^{1,3}=0.7$. The simulation results under Strategy 1 are plotted in Figure \ref{strategy1}.
\begin{figure}[tp!]
  \centering
 \hspace{-1pt}\includegraphics[width=0.23\textwidth]{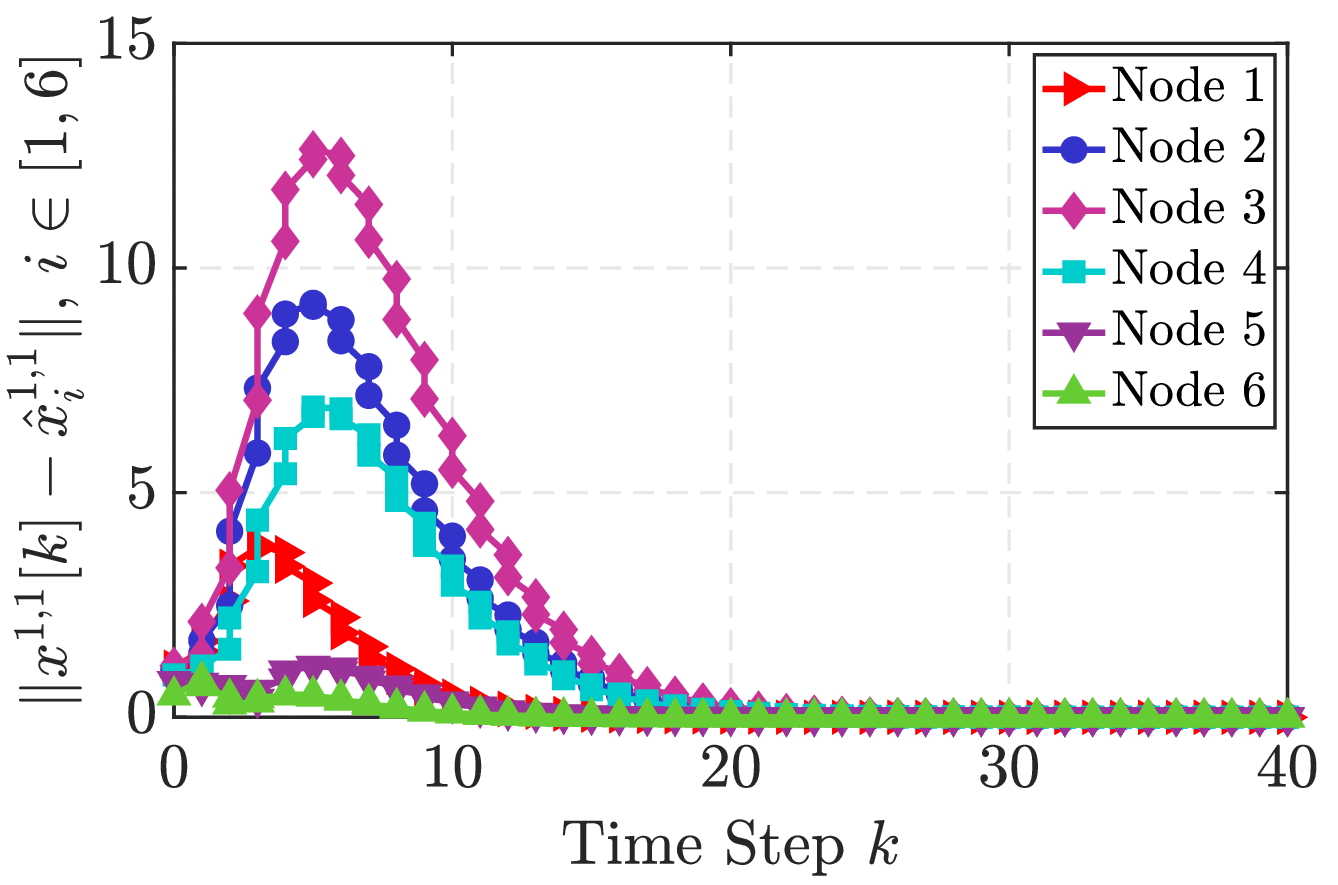}
  \hspace{-1pt}\includegraphics[width=0.23\textwidth]{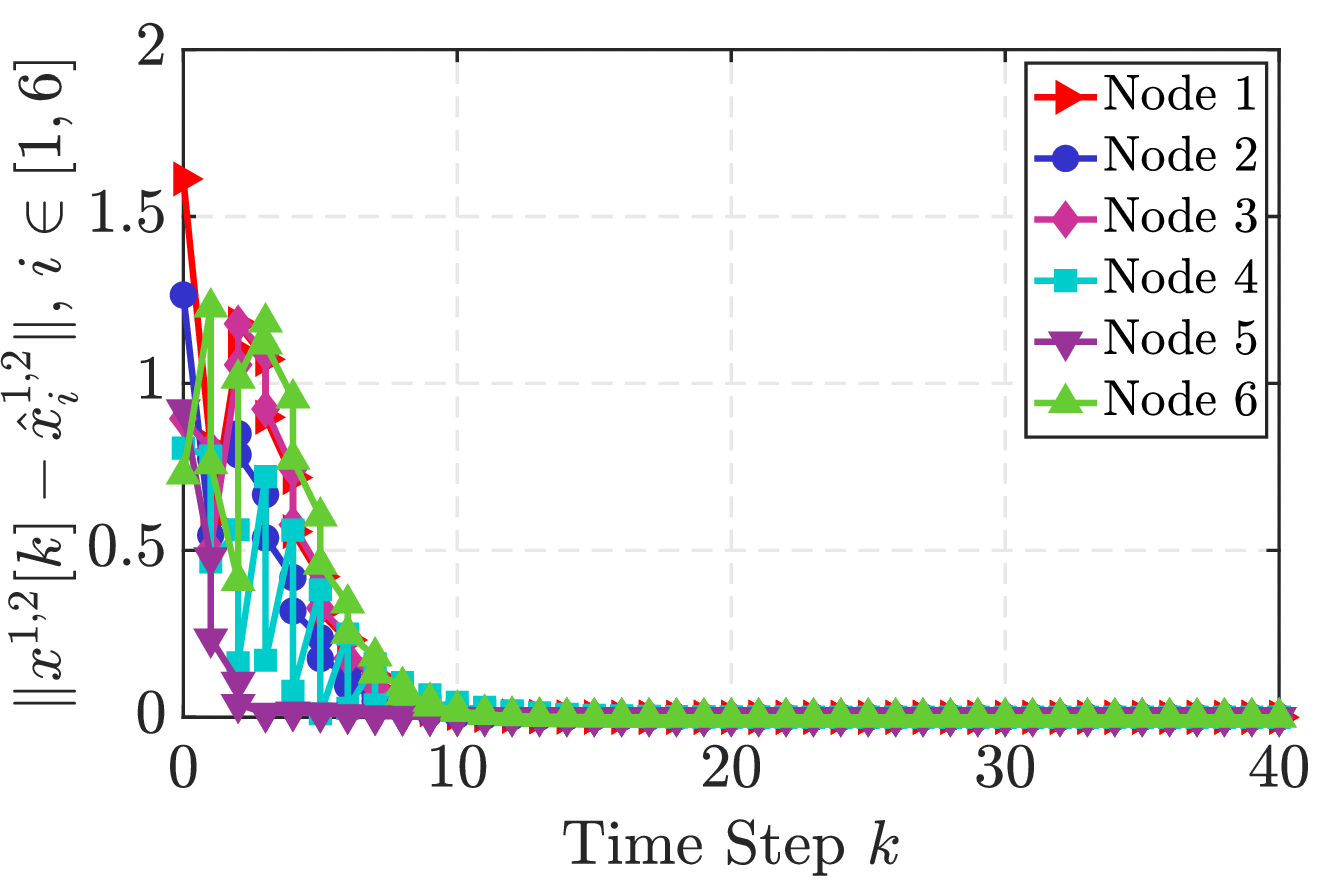}
   \hspace{-1pt}\includegraphics[width=0.23\textwidth]{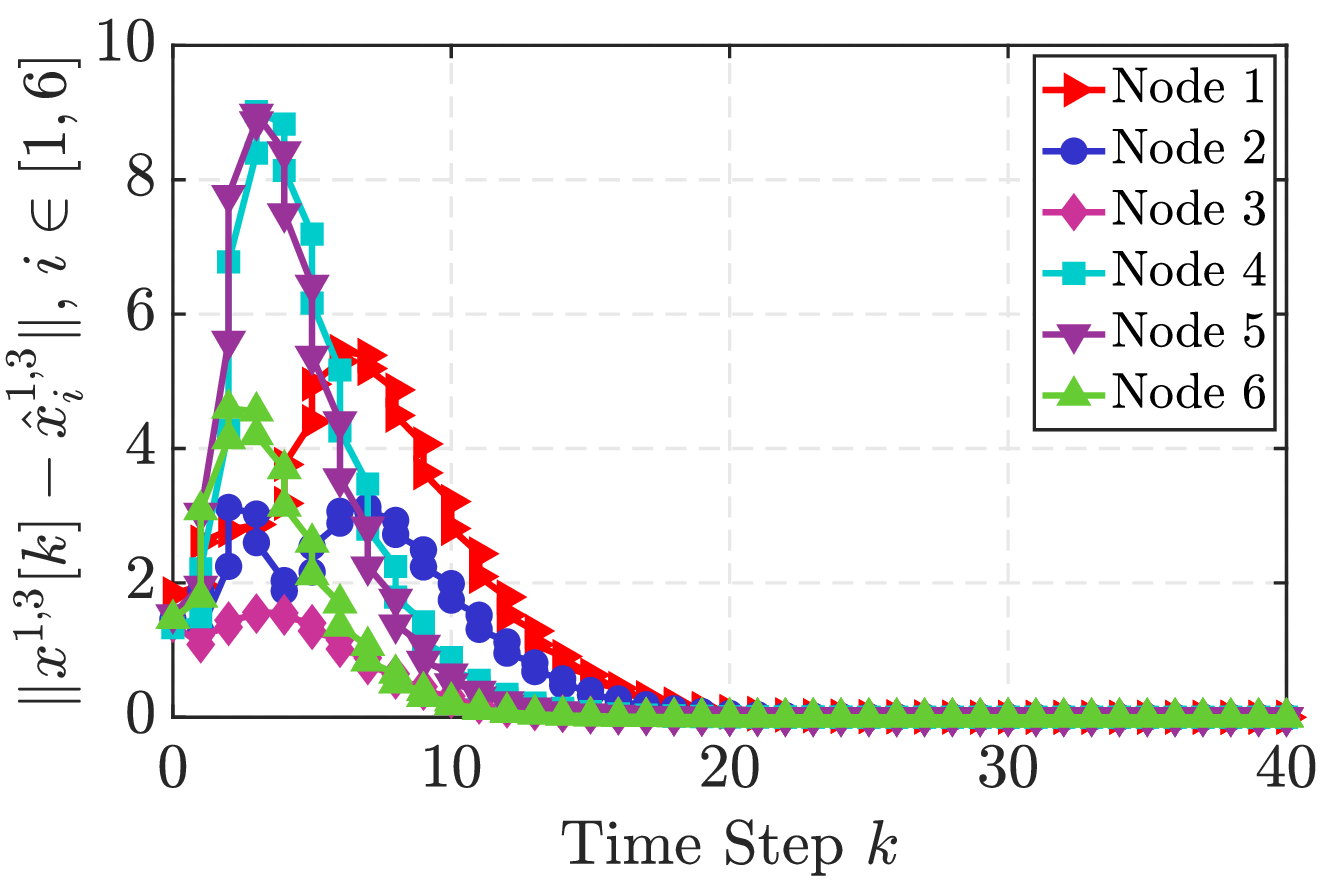}
  \caption{Norms of the estimation errors of all nodes under Strategy 1.}\label{strategy1}
\end{figure}

\textbf{Strategy 2.} {\color{black}Clearly, since the necessary and sufficient conditions for the problem solvability under Strategy 1 are satisfied, 
also those required  to apply Strategy 2 are met.
Following the procedure in Section \ref{sec.strategy2}, in order to obtain that the estimates of the substates $x_{i,d}(t)$, $i\in[1,6]$, we can  use the same Luenberger observers \eqref{eq.luenberger} derived under Strategy 1 to estimate the vectors $z_{i,d}(t)$, $i\in[1,6]$, and then single out the 
portions corresponding to miniblocks that are entirely observable for the various nodes. 

As far as the vectors  $x_{i,u}(t)$, $i\in[1,6]$, are concerned, each of them   is larger than the corresponding $z_{i,u}(t)$ adopted in Strategy 1, except for $x_{5,u}(t)$ that coincides with $z_{5,u}(t)$, since $5\not\in {\mathcal V}_2^{1,h_1}$ for any $h_1\in [1,3]$, namely no miniblock is partially observable for node $5$.
The  computation of the spectra of the Laplacians $\mathcal{L}^{1, h_1}, h_1\in [1,3],$ leads to}
$\sigma(\mathcal{L}^{1, 1})=\{1.8,0.3704, 2.4296, $ $1.2\}$, 
$\sigma(\mathcal{L}^{1, 2})= \{ 0.8,1, 1.8, 0.9\}$, 
$\sigma(\mathcal{L}^{1, 3})= \{ 0.3534,$ $2.5466, 1.8,1.2\}$. 
It follows from \eqref{eq.mainineq.aug} that for each $h_1\in[1,3]$, we only need to  solve $4$ inequalities to obtain $k^{1,h_1}$. This implies that compared to Strategy 1,  the number of inequalities to be verified in Strategy 2 is reduced. 
Specifically,  $0<k^{1,1}<0.8232$, $0<k^{1,2}<1.1111$, $0<k^{1,1}<0.7854$. We have set $k^{1,1}=0.5$, $k^{1,2}=1$, $k^{1,3}=0.7$. 
The norms of the estimation errors under Strategy 2 are depicted in Figure  \ref{strategy2}. 
\begin{figure}[htp!]
  \centering
 \hspace{-1pt}\includegraphics[width=0.23\textwidth]{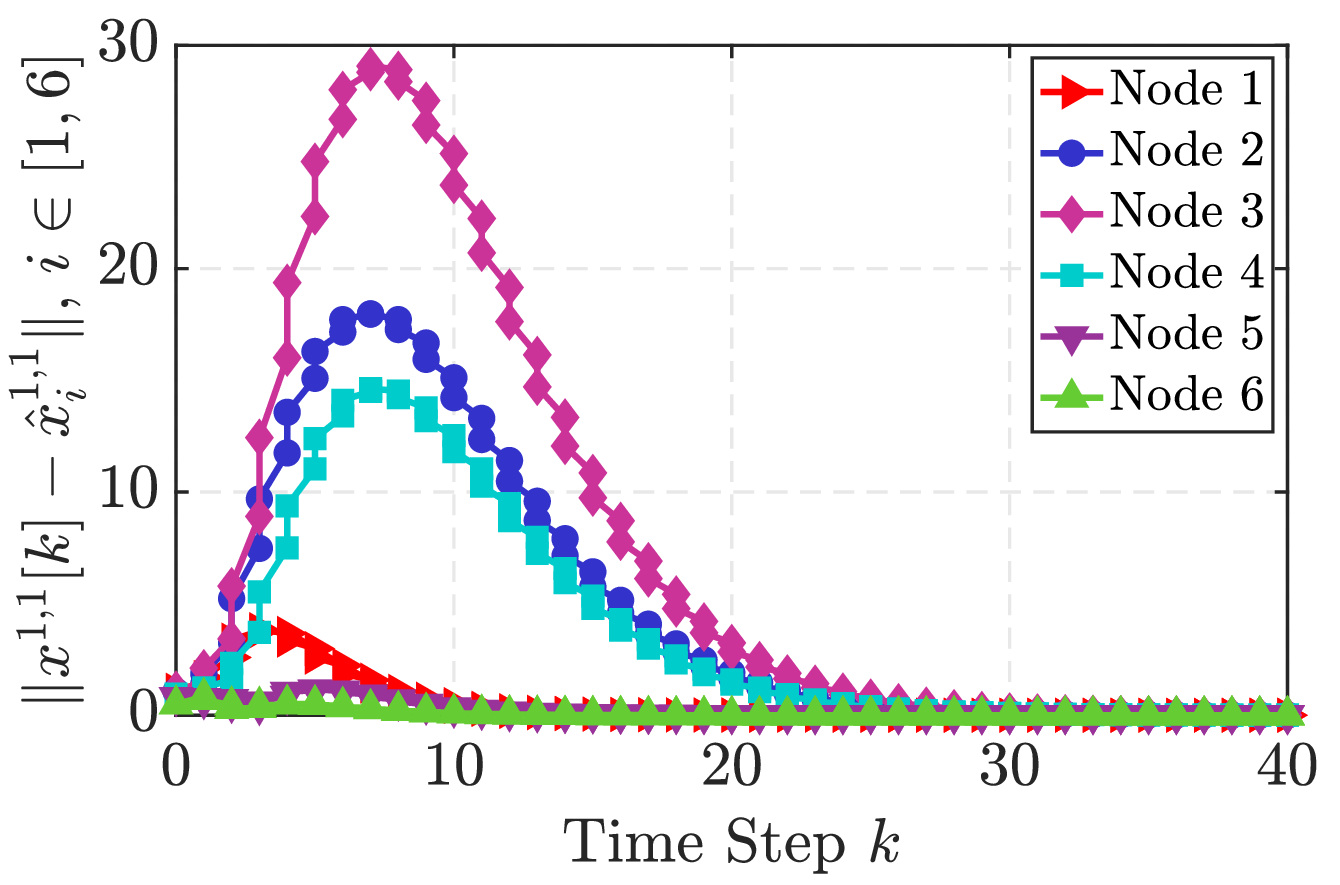}
  \hspace{-1pt}\includegraphics[width=0.23\textwidth]{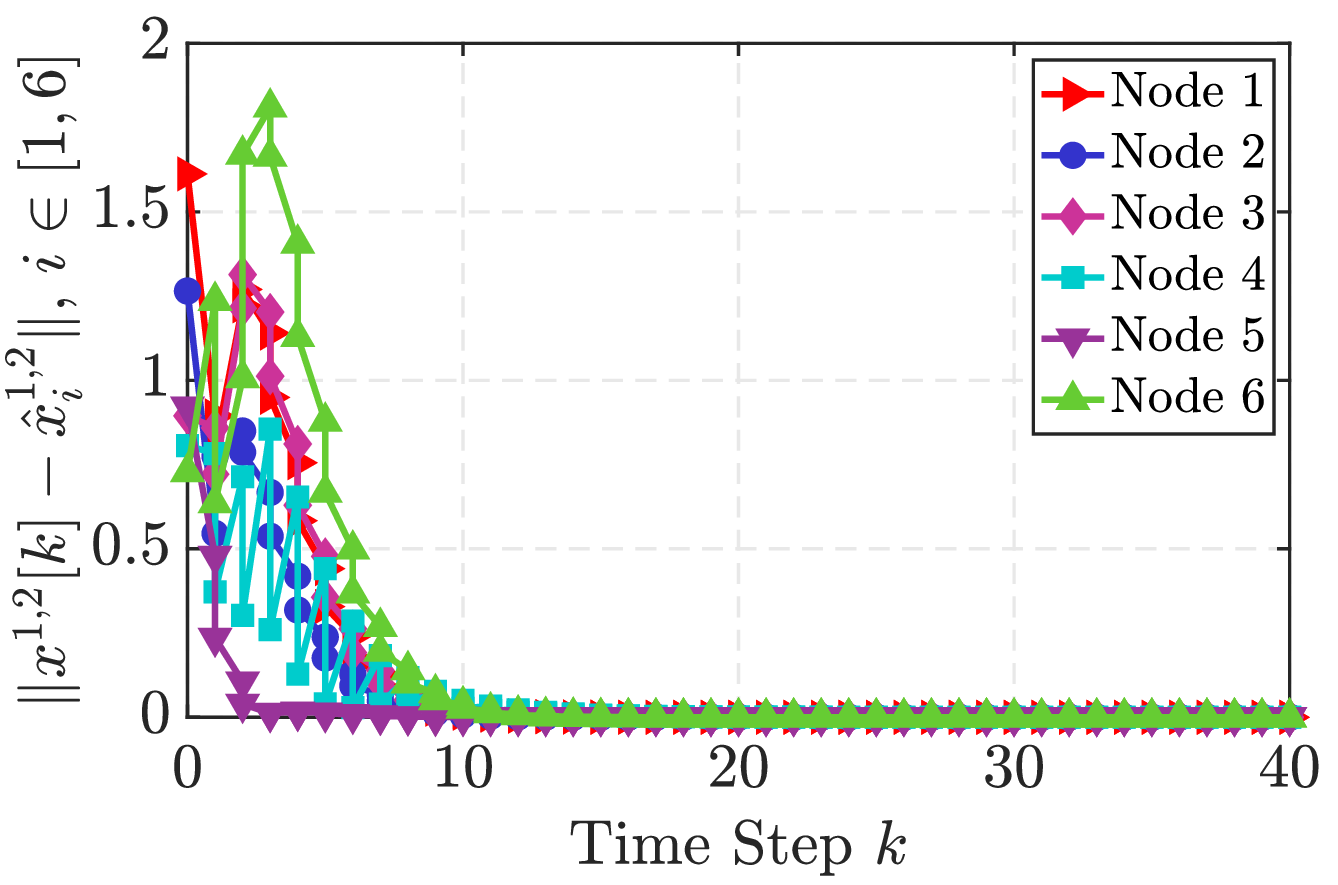}
   \hspace{-1pt}\includegraphics[width=0.23\textwidth]{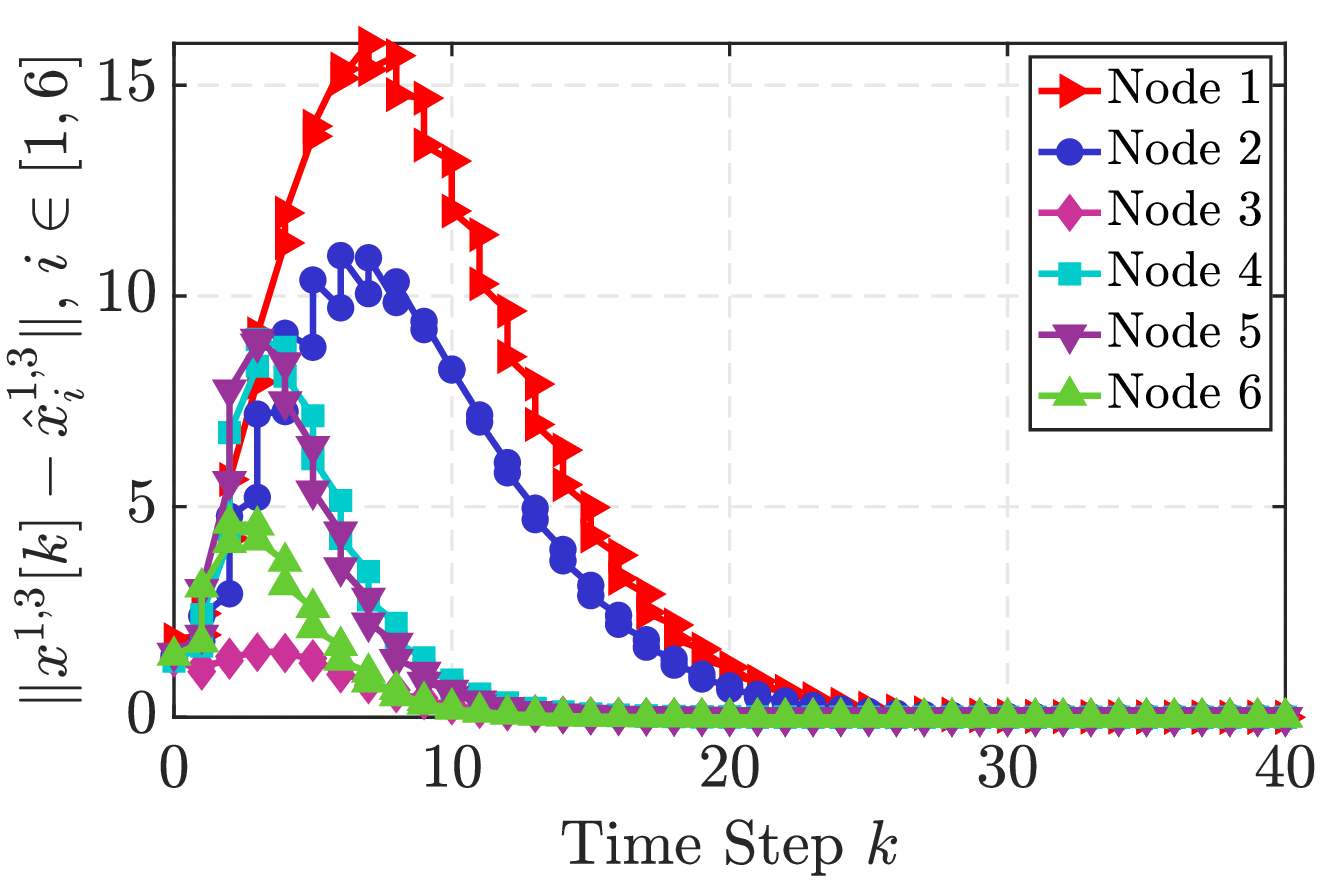}
  \caption{Norms of the estimation errors of all nodes under {\color{black}Strategy 2}.}\label{strategy2}
\end{figure}

\section{Conclusions}\label{sec.conclusion}
In this paper, we have proposed two strategies for distributed state estimation  that leverage both the Jordan canonical form of the system matrix $A$ and the  Kalman detectability form of each pair $(A,C_i)$. The first strategy proposes local observers whose dimension matches that of the original state vector, 
but imposes stringent solvability conditions. The second scheme, in contrast, can be implemented under weaker requirements, but results in observers of increased (augmented) order.
For both methods, we have derived necessary and sufficient conditions for their feasibility, expressed in terms
of the same set of inequalities but imposed on different sets of eigenvalues (of certain principal submatrices of the network Laplacian).  Compared to \cite{GaoYang}, the proposed approaches provide greater flexibility in the choice of coupling gains and impose less restrictive solvability conditions.

In the special case of an undirected communication graph, the feasibility conditions of the first strategy coincide with the simpler ones (see \eqref{eq.mainineq.aug}) derived for the second strategy. Hence, the first strategy is always preferable, as it leads to non-augmented observers.

Finally, we note that both proposed strategies rely on local coordinate permutations, and the gain matrices $L_{i,d}$ can be designed locally. In contrast, the coupling gains $k^{\ell, h_\ell}$ must be determined at a global level.

\appendix
\section*{Appendix}
\renewcommand{\thetheorem}{A.\arabic{theorem}} 
\setcounter{theorem}{0}
\setcounter{equation}{0}
\renewcommand{\theequation}{A.\arabic{equation}}

\begin{lemma}
\label{lem.1}
Given a Jordan miniblock $A\in {\mathbb R}^{d\times d}$ corresponding to some eigenvalue $\lambda\in {\mathbb R}, |\lambda|\ge 1$, a   matrix $L \in {\mathbb R}^{\bar N\times \bar N}$, a  matrix $S\doteq{\rm diag}\{S_i\}_{i\in[1,\bar N]}$, where $S_i \doteq \begin{bmatrix} I_{n_i} & {\mathbb 0}\end{bmatrix} \in {\mathbb R}^{n_i \times d}$, $1 \le n_i \le d$, is a selection matrix, and a scalar $k\in {\mathbb R}$, the following facts are equivalent.
\begin{itemize}
    \item[i)] The matrix 
$M\doteq
S  
\left(\left(I_{\bar N} - k L\right)\otimes A \right)
S^\top$ is Schur, i.e., all its   eigenvalues   have moduli smaller than $1$.
    \item[ii)] 
    For every $\mu \in\sigma(S(L\otimes I_d)S^\top)$, one has
    \begin{equation}
        \label{eq.diseq}
    |1 - k \mu|<\frac{1}{|\lambda|}.
    \end{equation}
\end{itemize}
\end{lemma}
\begin{proof}
To simplify the notation, set  $\Gamma\doteq\left(I_{\bar N} - k L\right)$.\\
 We  apply a permutation matrix $P$ to the rows of $S$ 
 in such a way that  the matrix $PMP^\top = (PS)  
\left(\left(I_{\bar N} - k L\right)\otimes A \right)
(PS)^\top$  becomes block triangular, and hence its spectrum is easy to determine. To this end, 
let us define $r_i\doteq {\rm card}(\{j\in[1,\bar N]:n_j\ge i\})$, $i\in[1,d]$, where $n_j$ is the number of rows of the block $S_j$. This is equivalent to saying that $r_i$ is the number of indices $j\in [1,\bar N]$ for which the selection matrix $S_j$ includes the row vector ${\mathbb e}_i^\top$. In general, $0\le r_i\le \bar N$ for every $i\in[1,d]$. Note that $\bar N=r_1\ge r_2\ge\cdots\ge r_d\ge0$. Moreover, $r_d\ge1$ if and only if there exists $j\in [1,\bar N]$ such that\footnote{This is the case when, if we assume $[1,\bar N]=\mathcal{V}_1 \cup\mathcal{V}_2 $, we have $\mathcal{V}_1 \ne \emptyset$. 
%Note that $r_i$ is the number of agents for which the $i$th component of the subvector $x^\ell,h_\ell$  is not detectable.
} $n_j=d$ (equivalently $S_j = I_d$). Finally, note that $\sum_{i=1}^{\bar N} n_i= \sum_{i=1}^d r_i$.
For every $i\in [1,d]$ such that $r_i>0$, let $\Tilde{S}_{r_i}\in \mathbb{R}^{r_i\times \bar N}$ be the selection matrix  
whose rows are the canonical vectors ${\mathbb e}_j^\top$ corresponding 
to the  indices in the set $\{ j\in [1,{\bar N}]: n_j \ge i\}$ (of cardinality $r_i$).  
Note that $r_1={\bar N}$, hence $\Tilde{S}_{r_1}=I_{\bar N}$. Moreover, if $r_i=r_{h}$  then $\tilde S_{r_i}= \tilde S_{r_h}$ (and this makes
the adopted notation consistent).
\\
Let  $P$ be  the permutation matrix such that
 \begin{align*}
    PS= 
    \begin{bmatrix}
           I_{\bar N}\otimes\mathbb{e}_1^\top \cr
         \Tilde{S}_{r_2}\otimes\mathbb{e}_2^\top
         \cr
         \vdots\cr
         \Tilde{S}_{r_d}\otimes\mathbb{e}_d^\top 
\end{bmatrix},
\end{align*}
and hence
$(PS)^\top =
     \begin{bmatrix}
          I_{\bar N}\otimes\mathbb{e}_1&
          \Tilde{S}_{r_2}^\top\otimes\mathbb{e}_2&
          \cdots&
          \Tilde{S}_{r_d}^\top\otimes\mathbb{e}_d 
 \end{bmatrix}.$
\\ The $(i,j)$th block of $(PS)(\Gamma\otimes A)(PS)^\top$ (of size $r_i\times r_j$) coincides with
$(\tilde S_{r_i}\otimes {\mathbb e}_i^\top) (\Gamma \otimes A) (\tilde S_{r_j}^\top\otimes {\mathbb e}_j) = 
(\Tilde{S}_{r_i} \Gamma\Tilde{S}_{r_j}^\top) \otimes (\mathbb{e}_i^\top A\mathbb{e}_j) = (\Tilde{S}_{r_i} \Gamma\Tilde{S}_{r_j}^\top) [A]_{i,j}.$
By the upper triangular structure of a Jordan miniblock, it follows that $[A]_{i,j} =0$ for $i> j$, while $[A]_{i,i}=\lambda$. 
Therefore
 the matrix $PMP^\top= (PS)(\Gamma\otimes A)(PS)^\top$ takes the following upper block triangular form:
\begin{align}\label{eq.blk.triag}
    PMP^\top= \begin{bmatrix}
   \lambda \Gamma &\star&\cdots&\star\\
    {\mathbb 0} & \lambda \Tilde{S}_{r_2} \Gamma\Tilde{S}_{r_2}^\top &\star&\vdots\\
    \vdots&\vdots&\ddots&\star\\
     {\mathbb 0}& {\mathbb 0}&\cdots& \lambda\Tilde{S}_{r_d} \Gamma\Tilde{S}_{r_d}^\top\\
    \end{bmatrix}.
\end{align}
where $\star$ denotes a entry whose value is not relevant.\\
It follows that 
 \begin{align*}
\sigma(M)&=
\sigma\left(PMP^\top\right)\\
&=\sigma\left(\begin{bmatrix}
   \lambda \Gamma &\star&\cdots&\star\\
    {\mathbb 0} & \lambda \Tilde{S}_{r_2} \Gamma\Tilde{S}_{r_2}^\top &\star&\vdots\\
    \vdots&\vdots&\ddots&\star\\
     {\mathbb 0}& {\mathbb 0}&\cdots& \lambda\Tilde{S}_{r_d} \Gamma\Tilde{S}_{r_d}^\top\\
    \end{bmatrix}\right)\\
    &=\sigma\left(\begin{bmatrix}
   \lambda \Gamma &\mathbb{0}&\cdots&\mathbb{0}\\
    {\mathbb 0} & \lambda \Tilde{S}_{r_2} \Gamma\Tilde{S}_{r_2}^\top &\mathbb{0}&\vdots\\
    \vdots&\vdots&\ddots&\mathbb{0}\\
     {\mathbb 0}& {\mathbb 0}&\cdots& \lambda\Tilde{S}_{r_d} \Gamma\Tilde{S}_{r_d}^\top\\
    \end{bmatrix}\right)\\
    &=\sigma\left(P^\top\begin{bmatrix}
   \lambda \Gamma &\mathbb{0}&\cdots&\mathbb{0}\\
    {\mathbb 0} & \lambda \Tilde{S}_{r_2} \Gamma\Tilde{S}_{r_2}^\top &\mathbb{0}&\vdots\\
    \vdots&\vdots&\ddots&\mathbb{0}\\
     {\mathbb 0}& {\mathbb 0}&\cdots& \lambda\Tilde{S}_{r_d} \Gamma\Tilde{S}_{r_d}^\top\\
    \end{bmatrix}P\right)\\
    &=\sigma\left(S\left(\Gamma\otimes\lambda I_d\right)S^\top\right)\\
    &=\sigma\left(S\left((I_{\bar N}-kL)\otimes\lambda I_d\right)S^\top\right).
\end{align*}
By the properties of the Kronecker product we have 
\begin{align*}
     S\left((I_{\bar N}-kL)\otimes\lambda I_d\right)S^\top=\lambda\left(I-k S\left(L\otimes I_d\right)S^\top\right),
\end{align*}
whose eigenvalues are $\{\alpha \in {\mathbb C}: \alpha= \lambda (1-k \mu), \exists \mu\in\sigma(S\left(L\otimes I_d\right)S^\top)\}$. Therefore $M$ is Schur if and only if $ii)$ holds.
\end{proof}

\begin{corollary}\label{cor.dir}
    Given a Jordan miniblock $A\in {\mathbb R}^{d\times d}$ corresponding to some eigenvalue $\lambda\in {\mathbb R}, |\lambda|\ge 1$, a   matrix $L \in {\mathbb R}^{{\bar N}\times {\bar N}}$, and a scalar $k\in {\mathbb R}$, the following facts are equivalent.
\begin{itemize}
    \item[i)] The matrix 
$M\doteq 
\left(\left(I_{\bar N} - k L\right)\otimes A \right)$ is Schur.
    \item[ii)] The inequality \eqref{eq.diseq} holds for every $\mu \in\sigma(L)$.
\end{itemize}
\end{corollary}
\begin{proof}
        The proof   immediately follows from Lemma \ref{lem.1}, by setting $S=I_{\bar N d}$ and using the properties of the spectrum of a Kronecker product.
\end{proof}
\begin{corollary}\label{cor.und}
    Given a Jordan miniblock $A\in {\mathbb R}^{d\times d}$ corresponding to some eigenvalue $\lambda\in {\mathbb R}, |\lambda|\ge 1$, a  symmetric  and positive definite matrix $L=L^\top \in {\mathbb R}^{{\bar N}\times {\bar N}}$, 
    with  $0 < \mu_m \le \mu_M$ being the minimum and maximum (real) eigenvalues of $L$, 
    a  matrix $S\doteq{\rm diag}\{S_i\}_{i\in[1,{\bar N}]}$, where $S_i = \begin{bmatrix} I_{n_i} & {\mathbb 0}\end{bmatrix} \in {\mathbb R}^{n_i \times d}$, $1 \le n_i \le d$, 
    and a scalar $k\in {\mathbb R}$,
    the following facts are equivalent.
\begin{itemize}
    \item[i)] The matrix   
$M\doteq
S  
\left(\left(I_{\bar N} - k L\right)\otimes A \right)
S^\top$ is Schur.
    \item[ii)] 
    The matrix $\lambda (I_{\bar N} - k L)$ is Schur.
    \item[iii)]  $k\in {\mathbb R}$ satisfies 
    \begin{align*}
k\in&\left(\frac{1}{\mu_{m}}\left(1-\frac{1}{|\lambda|}\right),\frac{1}{\mu_{M}}\left(1+\frac{1}{|\lambda|}\right)\right).
    \end{align*}
    \end{itemize} 
     Therefore there exists $k\in {\mathbb R}$ such that any of the above three equivalent conditions 
 holds if and only if  
 \begin{equation}
    \frac{\mu_{M}}{\mu_{m}}  <\frac{|\lambda| + 1}{|\lambda| -1}.
    \label{eq.ineq}
    \end{equation}
\end{corollary}
\begin{proof} $i)\Leftrightarrow ii)$\ 
   By      \eqref{eq.blk.triag} in the proof of Lemma \ref{lem.1}, the spectrum $\sigma(M)=\sigma\left(PMP^\top\right) $ can be written as 
\begin{align}
\sigma(M)&=
\sigma(\lambda \Gamma) \cup \left(\bigcup_{i=2}^d 
\sigma\left(\lambda \Tilde{S}_{r_i} \Gamma\Tilde{S}_{r_i}^\top\right)\right).
\label{eq.spectrumsplit}
\end{align}
Therefore if $M$ is Schur then $\lambda 
\Gamma = \lambda (I_N -k L)$ is Schur. On the other hand,
since $\Gamma$ is symmetric and all matrices $\Tilde{S}_{r_i} \Gamma\Tilde{S}_{r_i}^\top$ are principal submatrices of $\Gamma$, it follows from the Cauchy Interlacing Theorem \cite{Hwang01022004} that each eigenvalue $\beta$ of $\Tilde{S}_{r_i} \Gamma\Tilde{S}_{r_i}^\top$ satisfies $\min \{\alpha\in \sigma(\lambda \Gamma)\} \le \beta \le \max \{\alpha\in \sigma(\lambda \Gamma)\}$. This ensures that as soon as $\lambda \Gamma$ is Schur, then also $M$ is Schur. In this way we have proved $i)\Leftrightarrow ii)$.

$ii)\Leftrightarrow iii)$ The matrix $\lambda (I_{\bar N} -k L)$ is Schur if and only if for every $\mu\in \sigma(L)$ 
$$ |1 - k \mu|<\frac{1}{|\lambda|},$$
 or, equivalently,
for every $\mu \in \sigma(L)$
\begin{align}
 1-\frac{1}{|\lambda|}<k \mu<1+\frac{1}{|\lambda|}.
 \label{ineq.mu}
\end{align}
Since $0 <  \mu_m \le  \mu \le \mu_M$ and $1-\frac{1}{|\lambda|} \ge 0$, it follows that  $k>0$.
Moreover, inequality \eqref{ineq.mu}   holds for every $\mu \in \sigma(L)$  if and only if it holds for both $\mu_m$ and $\mu_M$, namely 
\begin{align*}
    \begin{cases}
        1-\frac{1}{|\lambda|}<k \mu_m<1+\frac{1}{|\lambda|}\\
        1-\frac{1}{|\lambda|}<k \mu_M<1+\frac{1}{|\lambda|}
    \end{cases}
\end{align*}
or, equivalently,
\begin{align}
    \begin{cases}
        \frac{1}{\mu_m} \left(1-\frac{1}{|\lambda|}\right) <k < \frac{1}{\mu_m} \left(1+\frac{1}{|\lambda|}\right)\\
         \frac{1}{\mu_M} \left(1-\frac{1}{|\lambda|}\right)<k  <\frac{1}{\mu_M} \left( 1+\frac{1}{|\lambda|}\right).
    \end{cases}
    \label{eq.2ineq}
\end{align}
But since
\begin{align*}
  \frac{1}{\mu_M}\left(1-\frac{1}{|\lambda|}\right)&\le  \frac{1}{\mu_m}\left(1-\frac{1}{|\lambda|}\right)\\
  \frac{1}{\mu_M} \left(1+\frac{1}{|\lambda|} \right)&\le\frac{1}{\mu_m}\left(1+\frac{1}{|\lambda|}\right),
\end{align*}
$k$   satisfies the  inequalities \eqref{eq.2ineq} if and only if \begin{align}\label{eq.ineq*}
\frac{1}{\mu_m}\left(1-\frac{1}{|\lambda|}\right)<  k < \frac{1}{\mu_M} \left(1+\frac{1}{|\lambda|} \right).
\end{align}
holds.

Clearly, there exists $k\in {\mathbb R}$ such that \eqref{eq.ineq*} holds (i.e., $iii)$ holds) if and only if 
\eqref{eq.ineq} holds.
\end{proof}

\bibliographystyle{model5-names} 
\bibliography{biblio}

\begin{thebibliography}{20}
\expandafter\ifx\csname natexlab\endcsname\relax\def\natexlab#1{#1}\fi
\providecommand{\bibinfo}[2]{#2}
\ifx\xfnm\relax \def\xfnm[#1]{\unskip,\space#1}\fi
%Type = Article
\bibitem[{Disarò et~al.(2025)Disarò, Fattore \& Valcher}]{Disaro_TAC2025}
\bibinfo{author}{Disarò, G.}, \bibinfo{author}{Fattore, G.}, \&
  \bibinfo{author}{Valcher, M.~E.} (\bibinfo{year}{2025}).
\newblock \bibinfo{title}{Distributed state estimation for discrete-time lti
  systems in the presence of unknown inputs}.
\newblock {\it \bibinfo{journal}{IEEE Transactions on Automatic Control}\/},
  (pp. \bibinfo{pages}{1--15}).
%Type = Inproceedings
\bibitem[{Fattore et~al.(2026)Fattore, Valcher, Gao \& Yang}]{Fattore_ECC2026}
\bibinfo{author}{Fattore, G.}, \bibinfo{author}{Valcher, M.~E.},
  \bibinfo{author}{Gao, R.}, \& \bibinfo{author}{Yang, G.-H.}
  (\bibinfo{year}{2026}).
\newblock \bibinfo{title}{Distributed state estimation of discrete-time lti
  systems via jordan canonical representation}.
\newblock \bibinfo{note}{To be presented at the 24th European Control
  Conference (ECC), July 7--10, 2026, Reykjavík, Iceland}.
%Type = Article
\bibitem[{Fioravanti et~al.(2024)Fioravanti, Makridis, Oliva, Vrakopoulou \&
  Charalambous}]{FioravantiTAC2024}
\bibinfo{author}{Fioravanti, C.}, \bibinfo{author}{Makridis, E.},
  \bibinfo{author}{Oliva, G.}, \bibinfo{author}{Vrakopoulou, M.}, \&
  \bibinfo{author}{Charalambous, T.} (\bibinfo{year}{2024}).
\newblock \bibinfo{title}{Distributed estimation and control for lti systems
  under finite-time agreement}.
\newblock {\it \bibinfo{journal}{IEEE Transactions on Automatic Control}\/},
  {\it \bibinfo{volume}{69}\/}, \bibinfo{pages}{7909--7916}.
%Type = Article
\bibitem[{Gao \& Yang(2025)}]{GaoYang}
\bibinfo{author}{Gao, R.}, \& \bibinfo{author}{Yang, G.-H.}
  (\bibinfo{year}{2025}).
\newblock \bibinfo{title}{Distributed state estimation for discrete-time linear
  systems: A canonical decomposition-based approach}.
\newblock {\it \bibinfo{journal}{IEEE Transactions on Automatic Control}\/},
  {\it \bibinfo{volume}{70}\/}, \bibinfo{pages}{4825--4832}.
%Type = Article
\bibitem[{Han et~al.(2019)Han, Trentelman, Wang \& Shen}]{Han2019}
\bibinfo{author}{Han, W.}, \bibinfo{author}{Trentelman, H.~L.},
  \bibinfo{author}{Wang, Z.}, \& \bibinfo{author}{Shen, Y.}
  (\bibinfo{year}{2019}).
\newblock \bibinfo{title}{A simple approach to distributed observer design for
  linear systems}.
\newblock {\it \bibinfo{journal}{IEEE Transactions on Automatic Control}\/},
  {\it \bibinfo{volume}{64}\/}, \bibinfo{pages}{329--336}.
%Type = Book
\bibitem[{Horn \& Johnson(1985)}]{HornJohnson}
\bibinfo{author}{Horn, R.~A.}, \& \bibinfo{author}{Johnson, C.~R.}
  (\bibinfo{year}{1985}).
\newblock {\it \bibinfo{title}{Matrix Analysis}\/}.
\newblock \bibinfo{address}{Cambridge, GB}: \bibinfo{publisher}{Cambridge Univ.
  Press}.
%Type = Article
\bibitem[{Huang et~al.(2022)Huang, Li \& Wu}]{ReviewDIstrObs}
\bibinfo{author}{Huang, S.}, \bibinfo{author}{Li, Y.}, \& \bibinfo{author}{Wu,
  J.} (\bibinfo{year}{2022}).
\newblock \bibinfo{title}{Distributed state estimation for linear
  time-invariant dynamical systems: A review of theories and algorithms}.
\newblock {\it \bibinfo{journal}{Chinese Journal of Aeronautics}\/},  {\it
  \bibinfo{volume}{35}\/}, \bibinfo{pages}{1--17}.
%Type = Article
\bibitem[{Hwang(2004)}]{Hwang01022004}
\bibinfo{author}{Hwang, S.-G.} (\bibinfo{year}{2004}).
\newblock \bibinfo{title}{Cauchy's interlace theorem for eigenvalues of
  hermitian matrices}.
\newblock {\it \bibinfo{journal}{The American Mathematical Monthly}\/},  {\it
  \bibinfo{volume}{111}\/}, \bibinfo{pages}{157--159}.
%Type = Book
\bibitem[{Kailath(1980)}]{Kailath}
\bibinfo{author}{Kailath, T.} (\bibinfo{year}{1980}).
\newblock {\it \bibinfo{title}{Linear Systems}\/}.
\newblock \bibinfo{publisher}{Prentice Hall, Inc.}
%Type = Inproceedings
\bibitem[{Kamgarpour \& Tomlin(2008)}]{Tomlin}
\bibinfo{author}{Kamgarpour, M.}, \& \bibinfo{author}{Tomlin, C.}
  (\bibinfo{year}{2008}).
\newblock \bibinfo{title}{Convergence properties of a decentralized kalman
  filter}.
\newblock In {\it \bibinfo{booktitle}{Proceedings of the 47th IEEE Conference
  on Decision and Control}\/} (pp. \bibinfo{pages}{3205--3210}).
\newblock \bibinfo{address}{Cancun, Mexico}.
%Type = Article
\bibitem[{Mitra \& Sundaram(2018)}]{MitraSundaram}
\bibinfo{author}{Mitra, A.}, \& \bibinfo{author}{Sundaram, S.}
  (\bibinfo{year}{2018}).
\newblock \bibinfo{title}{Distributed observers for lti systems}.
\newblock {\it \bibinfo{journal}{IEEE Transactions on Automatic Control}\/},
  {\it \bibinfo{volume}{63}\/}, \bibinfo{pages}{3689--3704}.
%Type = Article
\bibitem[{Mitra \& Sundaram(2019)}]{MitraSundaram2019}
\bibinfo{author}{Mitra, A.}, \& \bibinfo{author}{Sundaram, S.}
  (\bibinfo{year}{2019}).
\newblock \bibinfo{title}{Byzantine-resilient distributed observers for lti
  systems}.
\newblock {\it \bibinfo{journal}{Automatica}\/},  {\it
  \bibinfo{volume}{108}\/}, \bibinfo{pages}{108487}.
%Type = Inproceedings
\bibitem[{Olfati-Saber(2005)}]{OS_distrKalman2}
\bibinfo{author}{Olfati-Saber, R.} (\bibinfo{year}{2005}).
\newblock \bibinfo{title}{Distributed kalman filter with embedded consensus
  filters}.
\newblock In {\it \bibinfo{booktitle}{Proceedings of the 44th IEEE Conference
  on Decision and Control}\/} (pp. \bibinfo{pages}{8179--8184}).
%Type = Inproceedings
\bibitem[{Olfati-Saber(2007)}]{OS_distrKalman}
\bibinfo{author}{Olfati-Saber, R.} (\bibinfo{year}{2007}).
\newblock \bibinfo{title}{Distributed kalman filtering for sensor networks}.
\newblock In {\it \bibinfo{booktitle}{Proceedings of the 46th IEEE Conference
  on Decision and Control}\/} (pp. \bibinfo{pages}{5492--5498}).
\newblock \bibinfo{address}{New Orleans, USA}.
%Type = Inproceedings
\bibitem[{Park \& Martins(2012{\natexlab{a}})}]{Nuno1}
\bibinfo{author}{Park, S.}, \& \bibinfo{author}{Martins, N.~C.}
  (\bibinfo{year}{2012}{\natexlab{a}}).
\newblock \bibinfo{title}{An augmented observer for the distributed estimation
  problem for lti systems}.
\newblock In {\it \bibinfo{booktitle}{Proceedings of the 2012 American Control
  Conference}\/} (pp. \bibinfo{pages}{6775--6780}).
\newblock \bibinfo{address}{Montreal, Canada}.
%Type = Inproceedings
\bibitem[{Park \& Martins(2012{\natexlab{b}})}]{Nuno2}
\bibinfo{author}{Park, S.}, \& \bibinfo{author}{Martins, N.~C.}
  (\bibinfo{year}{2012}{\natexlab{b}}).
\newblock \bibinfo{title}{Necessary and sufficient conditions for the
  stabilizability of a class of lti distributed observers}.
\newblock In {\it \bibinfo{booktitle}{Proceedings of the 51st IEEE Conference
  on Decision and Control}\/} (pp. \bibinfo{pages}{7431--7436}).
\newblock \bibinfo{address}{Maui, Hawaii}.
%Type = Article
\bibitem[{Park \& Martins(2017)}]{Nuno3}
\bibinfo{author}{Park, S.}, \& \bibinfo{author}{Martins, N.~C.}
  (\bibinfo{year}{2017}).
\newblock \bibinfo{title}{Design of distributed lti observers for state
  omniscience}.
\newblock {\it \bibinfo{journal}{IEEE Transactions on Automatic Control}\/},
  {\it \bibinfo{volume}{62}\/}, \bibinfo{pages}{561--576}.
%Type = Article
\bibitem[{Rego et~al.(2019)Rego, Pascoal, Aguiar \& Jones}]{Rego19}
\bibinfo{author}{Rego, F. F.~C.}, \bibinfo{author}{Pascoal, A.~M.},
  \bibinfo{author}{Aguiar, A.~P.}, \& \bibinfo{author}{Jones, C.~N.}
  (\bibinfo{year}{2019}).
\newblock \bibinfo{title}{Distributed state estimation for discrete-time linear
  time invariant systems: A survey}.
\newblock {\it \bibinfo{journal}{Annual Reviews in Control}\/},  {\it
  \bibinfo{volume}{48}\/}, \bibinfo{pages}{36--56}.
%Type = Article
\bibitem[{Wang et~al.(2024)Wang, Liu, Anderson \& Morse}]{Wang24}
\bibinfo{author}{Wang, L.}, \bibinfo{author}{Liu, J.},
  \bibinfo{author}{Anderson, B. D.~O.}, \& \bibinfo{author}{Morse, A.~S.}
  (\bibinfo{year}{2024}).
\newblock \bibinfo{title}{Split-spectrum based distributed state estimation for
  linear systems}.
\newblock {\it \bibinfo{journal}{Automatica}\/},  {\it
  \bibinfo{volume}{161}\/}, \bibinfo{pages}{111421}.
%Type = Article
\bibitem[{Yang et~al.(2022)Yang, Barboni, Rezaee \& Parisini}]{yang2022state}
\bibinfo{author}{Yang, G.}, \bibinfo{author}{Barboni, A.},
  \bibinfo{author}{Rezaee, H.}, \& \bibinfo{author}{Parisini, T.}
  (\bibinfo{year}{2022}).
\newblock \bibinfo{title}{State estimation using a network of distributed
  observers with unknown inputs}.
\newblock {\it \bibinfo{journal}{Automatica}\/},  {\it
  \bibinfo{volume}{146}\/}, \bibinfo{pages}{110631}.

\end{thebibliography}
\end{document}